\documentclass[%
 reprint,
superscriptaddress,
 amsmath,amssymb,
 aps,
]{revtex4-2}

\usepackage{graphicx}
\usepackage{dcolumn}
\usepackage{bm}

\usepackage{hyperref}

\usepackage{tikz}

\usepackage{graphicx}

\usepackage{array}

\usepackage{amsmath}
\usepackage{amsfonts}
\usepackage{amssymb}
\usepackage{amsthm}
\usepackage{mathtools}
\newtheorem{theorem}{Theorem}
\newtheorem{lemma}{Lemma}
\newtheorem{corollary}{Corollary}
\newtheorem{definition}{Definition}
\theoremstyle{definition}

\begin{document}

\preprint{APS/123-QED}

\title{Classical simulation of quantum circuits with noisy magic inputs}

\author{Jiwon Heo}
\affiliation{%
 Graduate School of Quantum Science and Technology, Korea Advanced Institute of Science and Technology~(KAIST)
}%

\author{Sojeong Park}
\affiliation{
 Department of Physics, Korea Advanced Institute of Science and Technology~(KAIST)
}

\author{Changhun Oh}
\email{changhun0218@gmail.com}
\affiliation{%
 Graduate School of Quantum Science and Technology, Korea Advanced Institute of Science and Technology~(KAIST)
}%
\affiliation{
 Department of Physics, Korea Advanced Institute of Science and Technology~(KAIST)
}

\date{\today}

\begin{abstract}
    Magic states are essential for universal quantum computation and are widely viewed as a key source of quantum advantage, yet in realistic devices they are inevitably noisy. In this work, we characterize how noise on injected magic resources changes the classical simulability of quantum circuits and when it induces a transition from classically intractable behavior to efficient classical simulation. We adopt a resource-centric noise model in which only the injected magic components are noisy, while the baseline states, operations, and measurements belong to an efficiently simulable family. Within this setting, we develop an approximate classical sampling algorithm with controlled error and prove explicit noise-dependent conditions under which the algorithm runs in polynomial time. Our framework applies to both qubit circuits with Clifford baselines and fermionic circuits with matchgate baselines, covering representative noise channels such as dephasing and particle loss. We complement the analysis with numerical estimates of the simulation cost, providing concrete thresholds and runtime scaling across practically relevant parameter regimes.
\end{abstract}

\maketitle


\section{Introduction}
A quantum computer is a computational device that harnesses fundamental quantum phenomena, such as quantum entanglement and superposition, to achieve computational speedups~\cite{Nielsen_Chuang}. Owing to these features, quantum computers are expected to outperform classical computers in certain tasks, such as integer factoring and the simulation of quantum systems' time evolution, a phenomenon often referred to as quantum computational advantage~\cite{Supremacy1, Supremacy2, Shor2, Boson_sampling, Gaussian_boson_sampling, IQP, RCS, Fermion_sampling}. 
Recently, many experiments have aimed to achieve quantum advantage via sampling-based tasks, such as random-circuit sampling and Gaussian boson sampling~\cite{RCS_experiment, GBS_experiment}.
Despite these advances, realizing practical quantum advantage remains challenging because current hardware platforms are subject to substantial physical noises~\cite{Classical_RCS1, Classical_RCS2, Classical_GBS1}.
Although quantum error correction promises to address the effect of noise~\cite{QEC1, QEC2}, building a fault-tolerant quantum computer remains beyond current technology.
Hence, understanding the power of noisy quantum devices hinges on clarifying how noise affects computational tasks.

More specifically, many studies have shown that typical physical noise progressively reduces quantum advantage and can eventually destroy it~\cite{Classical_RCS1, Classical_RCS2, Classical_BS1, Classical_BS2, Classical_GBS1, Noise_uniform,Noisy_gate1,Noisy_gate2,Noisy_gate3,Noisy_gate4,Noisy_arbitrary, Noisy_sampling_IQP,Noisy_sampling_RCS, Error1}.
A representative example is the depolarizing channel in qubit systems. 
When this noise channel is applied to each circuit layer, the system’s entropy increases and, eventually, the output state converges to the maximally mixed state, washing out any computational advantage~\cite{aharonov1996limitations}. 
For a better understanding of the effect of noise and pinpointing the regime where a quantum advantage is maintained, recent studies increasingly focus on characterizing the transition between classically intractable and tractable regimes as a function of noise strength~\cite{Classical_GBS1, Hardness_GBS1, Hardness_BS1, Hardness_BS2}.


Meanwhile, a standard route to universal quantum computation augments Clifford operations with a non-stabilizer resource, typically supplied in the form of magic states~\cite{Nielsen_Chuang,Fermion_computing,Fermion_computing2,Magic}.
In contrast, circuits restricted to stabilizer states, Clifford gates, and computational-basis measurements admit efficient classical simulation (e.g., via the Gottesman–Knill theorem)~\cite{Gottesman, Stabilizer_tableau}.
Magic states are therefore the essential ingredient that promotes Clifford circuits to universality and, more broadly, underpins proposals for quantum advantage beyond the stabilizer regime.
It is therefore natural to expect that when physical noise suppresses the non-stabilizer content~(``magic”) of these states, the corresponding circuit becomes easier to classically simulate.

In realistic devices, the preparation of ideal magic states is inevitably affected by hardware imperfections, resulting in noisy magic states with reduced computational power, a limitation commonly referred to as state-preparation error~\cite{SPAM1,SPAM2}. To mitigate this issue, magic-state distillation (MSD) protocols have been developed~\cite{Magic_distillation1,Magic_distillation2, Magic_distillation3}, which purify a small number of high-fidelity magic states from an ample supply of noisy ones. However, MSD incurs substantial overhead in both circuit depth and qubit resources, and remains one of the main bottlenecks for fault-tolerant architectures. 
Consequently, near-term experiments typically operate in a regime where magic states are noisy. Understanding their computational power, particularly determining when noise renders them efficiently classically simulable, is therefore crucial for assessing whether near-term quantum devices can exhibit rigorous quantum advantage. In this work, we address this question by analyzing the transition between quantum-hard and classically tractable regimes under realistic noise acting on magic states.

\subsection{Related works}
Over many decades, the classical simulability of quantum circuits has been extensively studied. A representative result is the Gottesman-Knill theorem for qubit stabilizer circuits~\cite{Gottesman, Stabilizer_tableau}, which implies that there is a classical algorithm for simulating a quantum scheme consisting of a stabilizer state, a Clifford operation, and the computational basis measurement. 
In fermionic settings, an analogous efficiently simulable model is provided by matchgate computation~(equivalently, fermionic linear optics), where efficient simulation follows from the underlying free-fermion structure~\cite{Matchgate,Matchgate2,Matchgate3}. More precisely, it is based on the fact that any Gaussian operator can be decomposed as nearest-neighbor~(NN) matchgates. Subsequently, Brod extended these simulability results to more general input and measurement settings by leveraging limited auxiliary resources~\cite{Classical_FLO1}.

Beyond these tractable families, substantial effort has been devoted to classically simulating circuits that include limited non-stabilizer or non-Gaussian resources, such as magic states~\cite{Bravyi,Stabilizer_rank,Gaussian_rank1,Gaussian_rank2, Bosonic_Gaussian_rank1,Bosonic_Gaussian_rank2, MPS1, MPS2, MPS3}.
One prominent direction is the rank-based~(decomposition) approach~\cite{Bravyi,Stabilizer_rank,Gaussian_rank1,Gaussian_rank2, Bosonic_Gaussian_rank1,Bosonic_Gaussian_rank2}. Bravyi and Gosset introduced an approximate classical simulation algorithm for Clifford-dominated circuits~\cite{Bravyi}. Their approach can be viewed as: (i) reducing non-Clifford gates to the injection of corresponding magic states via standard gadgets, and (ii) approximating those magic states by a small superposition of stabilizer states. 
Later, Bravyi et al. refined this framework by introducing quantitative notions such as stabilizer rank and stabilizer extent, together with fidelity-based performance guarantees, thereby strengthening bounds on the classical simulation cost~\cite{Stabilizer_rank}.
Related rank-based ideas have also been developed in fermionic and bosonic Gaussian settings, yielding notions such as Gaussian rank and its variants~\cite{Gaussian_rank1,Gaussian_rank2,Bosonic_Gaussian_rank1,Bosonic_Gaussian_rank2}.

Finally, since physical noise can degrade or even eliminate computational advantage, it is important to understand how noise acts on the non-stabilizer resources themselves. In the fermionic setting, de Melo et al. studied the effect of depolarizing noise on auxiliary resource states and related the onset of efficient classical simulation to whether the noisy state becomes convex-Gaussian, i.e., a convex mixture of fermionic Gaussian states~\cite{Fermion_error1}. They derived noise thresholds that separate regimes where the resource state is provably non-convex-Gaussian from regimes where an explicit convex-Gaussian decomposition exists, implying efficient classical simulation in the latter. Subsequently, Oszmaniec et al. expanded this line of work by providing an analytic characterization of convex-Gaussian states in the first nontrivial (four-mode) case and by enlarging the known classically simulable region under depolarizing noise~\cite{Fermion_error2}.


\begin{figure*}[t!]
    \centering
    \includegraphics[width=1.0\linewidth]{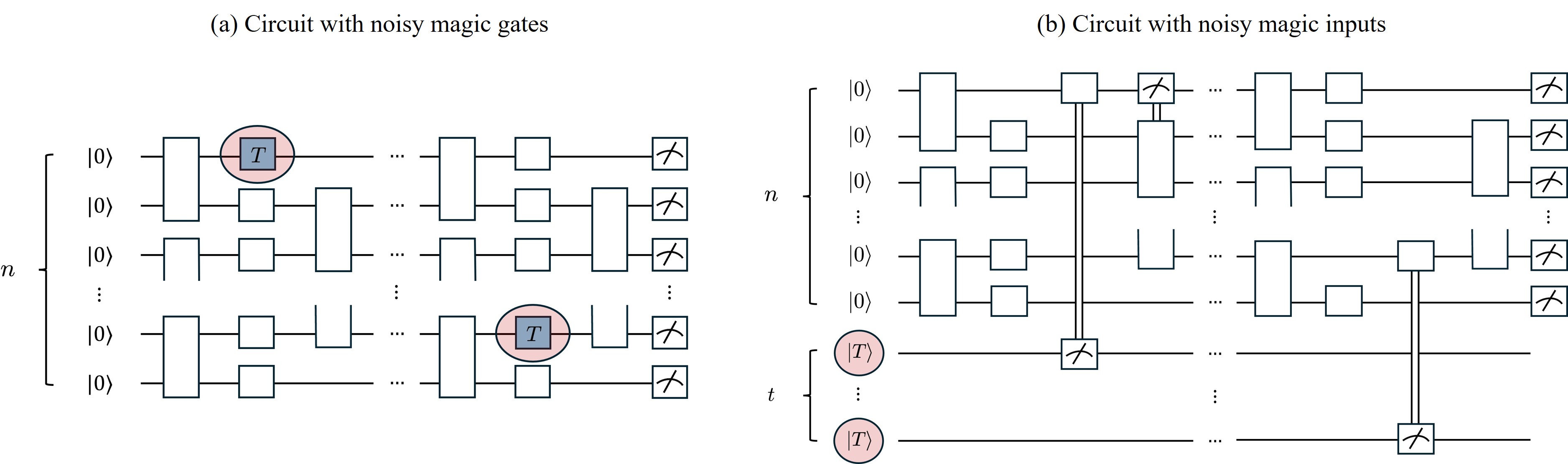}
    \caption{A qubit-based example of quantum circuits that permitted universal quantum computation with noisy magic components, which is our desired circuit. White boxes are Clifford operations, and measurements are the computational basis measurement (a) A $n$-qubit circuit with $t$ $T$ gates, whereas these suffer from physical noise. (b) Injecting $|T\rangle$ with adaptive measurement; however, it also suffers from physical noise. We mention that we also permit adaptive measurements in the injected circuit. }
    \label{fig:Scheme}
\end{figure*}

\subsection{Our contributions}

Although substantial work has investigated magic-state resources and the effects of noise, the computational power of \emph{noisy} magic states, and, in particular, the noise-dependent boundary between classically tractable and classically intractable behavior remains not fully understood.
A natural way to address this gap is to develop explicit classical simulation algorithms whose runtime can be analyzed as a function of the noise strength and resource parameters. Such an analysis clarifies how noise degrades computational power and pinpoints thresholds at which the simulation cost becomes polynomial, thereby delineating a classically simulable regime.

In this work, we study how physical noise acting on injected magic resources affects the computational power of otherwise efficiently simulable circuit families. We adopt a resource-centric noise model in which noise acts solely on the injected magic components, while the baseline free states, operations, and measurements are treated as ideal, allowing the degradation of the nonfree resource to be analyzed in isolation. This yields a clean characterization of the transition from quantum-hard to efficiently classically simulable behavior. As concrete instantiations, we consider qubit circuits with stabilizer baselines and fermionic linear-optical circuits with matchgate baselines, where the addition of suitable magic resources promotes universality.

Technically, we employ tools that systematically move all nonfree elements to the input-state side, reducing the problem to simulating a fixed free backbone driven by noisy resource states.
In the qubit setting, this is achieved via standard injection gadgets such as $T$-gadgetization~\cite{Gadget, Magic} (see Fig.~\ref{fig:Scheme}).
This perspective naturally captures regimes in which errors on the injected magic resources dominate, and it yields explicit noise thresholds together with analytic upper bounds on the advantage attainable with noisy magic resources.
Consequently, we identify parameter regimes in which noisy magic states still enable nontrivial quantum computation and regimes in which efficient classical simulation suffices.

To this end, we demonstrate how to instantiate our framework in concrete scenarios. In the qubit setting, we adopt the magic state $|H\rangle = \cos(\pi/8)|0\rangle + \sin(\pi/8)|1\rangle$ and consider dephasing noise. In the fermionic setting, we adopt the two-pair state $|\psi_4\rangle = (|0011\rangle + |1100\rangle)/\sqrt{2}$ as the injected resource and consider particle-loss and dephasing noise. For each example, we derive the simulation cost analytically and determine the noise-dependent boundary at which the cost becomes polynomial, thereby identifying the corresponding classically simulable regime~(see Fig.~\ref{fig:Total regime}). We emphasize that our method is formulated in a general setting; hence, it can be applied to other circuit families and noise models to obtain analogous computational boundaries.


In other words, our classical simulation algorithm targets sampling in the noisy-magic regime and relies on three ingredients. First, we exploit the fact that noise maps pure inputs to mixed states and use a Monte Carlo reduction: instead of evolving the full density operator, we sample a pure state from an ensemble decomposition and simulate that instance to generate a measurement outcome. Second, when the noise is sufficiently strong, the induced sampling distribution places negligible weight on instances containing many magic insertions; we therefore introduce a truncation scheme that discards high-magic samples while keeping the resulting bias under control. Third, conditioned on a sampled pure instance, we reduce the per-sample computational cost by leveraging stabilizer-decomposition techniques from prior work~\cite{Bravyi}, which are efficient whenever the retained instances contain only a small number of magic insertions. Taken together, these steps yield a sampling simulator that produces draws from a controlled approximation to the circuit’s output distribution with polynomial runtime once noise suppresses the high-magic tail below our derived threshold. The framework applies to qubit circuits with Clifford baselines and fermionic circuits with matchgate baselines, while treating other operations and measurements as noiseless to isolate the role of magic-state noise.

This paper is organized as follows: In Sec.~\ref{Sec:General framework}, we describe the general problem that we focus on, and propose an algorithm that simulates it classically. In Sec.~\ref{Sec:Qubit} and Sec.~\ref{Sec:Fermion}, we show how we apply our strategy to the various systems. We discuss the qubit-based case in Sec.~\ref{Sec:Qubit}, and the fermionic case in Sec.~\ref{Sec:Fermion}. Furthermore, we numerically depict the computational costs of our algorithm for each case by varying several variables, including the number of qubits, the noise rate, and the approximation error~(see Fig.~\ref{fig:Numerical}). Finally, we give discussions and the remaining open questions in Sec.~\ref{Sec:Discussion}.

\begin{figure*}[ht!]
    \centering
    \includegraphics[width=1.0\linewidth]{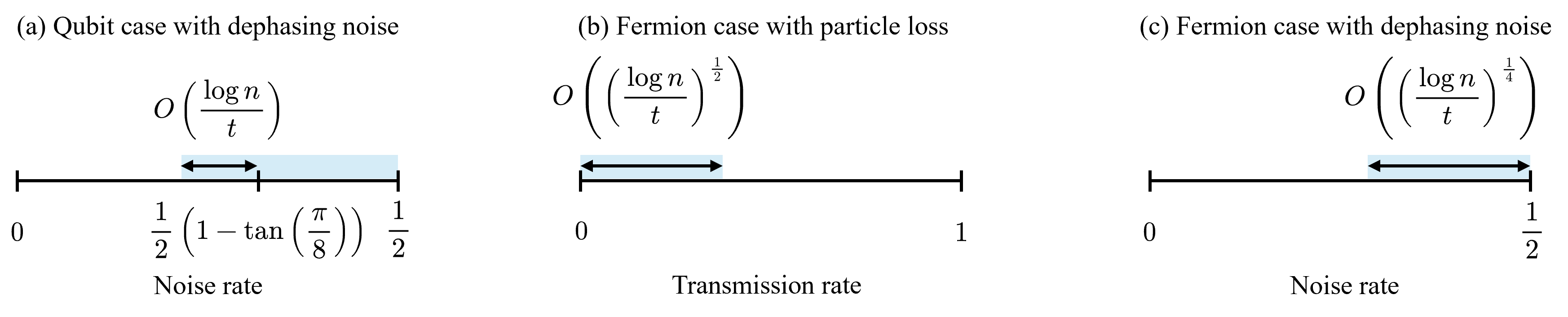}
    \caption{Classical simulable regions in each case (Blue boxes). For the dephasing noise cases, we only represent when the noise rate is in $[0, 1/2]$, because of symmetry. (a) In this case, the noisy magic states can be represented as a convex summation of stabilizer states when the noise rate is in $[(1-\tan(\pi/8))/2, 1/2]$, so that this interval is in the classical simulable region. For the remaining interval, the area of the classical simulable region is $O(\log n /t)$. (b) The area of the classical simulable region is $O(\sqrt{\log n /t})$. (c) The area of the classical simulable region is $O((\log n /t)^{1/4})$.
    }
    \label{fig:Total regime}
\end{figure*}

\section{General framework}\label{Sec:General framework}
\subsection{Problem description}\label{Sec:Problem}
As noted previously, circuits built solely from resourceless primitives are known to be efficiently classically simulable (e.g., via the Gottesman--Knill theorem for stabilizer circuits, and analogous methods for fermionic linear optics)~\cite{Gottesman, Stabilizer_tableau, Classical_FLO1, Classical_FLO2}. 
In contrast, injecting magic resources typically renders such simulations costly: a generic magic state must be expanded as a superposition of exponentially many resourceless states, and the simulation cost scales with the number of terms in the superposition. 
This observation is formalized by notions such as the stabilizer rank (and its fermionic/Gaussian analogues), which quantify how many resourceless states are required to represent a given resourceful state~\cite{Stabilizer_rank, Gaussian_rank2}.

At the same time, it is widely expected that noise can reduce the effective nonclassicality, or ``magic'', of the injected resource, potentially making a low-rank description possible, especially if one allows a small approximation error. 
This suggests a concrete route toward approximate classical simulation: if a noisy magic input can be well-approximately described by a superposition of only a few resourceless states, then one can classically simulate each resourceless branch (using, e.g., Gottesman-Knill-type methods) and combine them through phase-sensitive post-processing, thereby obtaining an efficient approximate sampler~\cite{Stabilizer_rank, Gaussian_rank2}.

Motivated by this intuition, we ask the following key question:
\emph{Can the noisy magic components in our scheme be replaced (up to a small error) by a low-rank superposition of resourceless states?}
Equivalently, does noise reduce the relevant (approximate) rank to a level that collapses our setting into only a few phase-sensitive resourceless computations?

To address this question, we first formalize a system-independent setting that isolates how noise acting on magic resources affects computational power. Throughout, we distinguish between resourceless components, which admit efficient classical simulation (e.g., Gottesman-Knill theorem), and magic components, which promote a baseline model to computational universality once supplied in sufficient quantity with adaptive measurements. On the qubit system, stabilizer states, Clifford operations, and computational-basis measurements are resourceless, whereas non-stabilizer states and non-Clifford gates are magic. On the fermionic system, Gaussian states and Gaussian evolutions with number measurements are resourceless, whereas non-Gaussian resources are magic. These examples serve only to ground the discussion; our definitions and results do not depend on a particular system.
The circuits of interest consist of a baseline built from resourceless operations and measurements together with a finite supply of injected magic states. Other input states are assumed to be resourceless, such as computational-basis states or Gaussian states. 
Hence, our starting point is a noiseless template with three parts:
\begin{itemize}
    \item Input: a tensor product of an $n$-qubit~(or mode) resourceless state $|0^{\otimes n}\rangle$ and $t$ copies of magic states $|\psi\rangle$, i.e., $|0^{\otimes n}\rangle\otimes|\psi^{\otimes t}\rangle$~($|\psi\rangle$ need not be a single qubit or mode state and $t = O(\text{poly}(n))$.
    \item Operation: an evolution $\hat{U}$ built entirely from resourceless gates (e.g., Clifford in the qubit system or Gaussian/matchgate evolutions in the fermion system). It consists of polynomially many single or two-qubit~(adaptive) Clifford gates in the qubit system (Or, it consists of polynomially many~(adaptive) NN matchgates in the fermionic system).
    \item Measurement: computational-basis (or number-basis) measurements, as appropriate for the system.
\end{itemize}

Following our motivation, we introduce noise only to the injected magic states while the operations in the baseline and the measurements are assumed noiseless. This resource-centric placement of noise reflects regimes where magic-state-preparation errors dominate and allows us to attribute any change in computational behavior directly to the degradation of magic.
Thus, the resulting input is a product of density operators, $|0^{\otimes n}\rangle\langle0^{\otimes n}|\otimes \hat{\rho}^{\otimes t}$, with $\hat{\rho}$ the noisy version of $|\psi\rangle$.
The overall scheme is illustrated in Fig.~\ref{fig:Scheme}.

Our computational task is sampling. Given a circuit from this family and a specification of the noise parameter, the goal is to generate samples from a distribution that approximates the noisy output distribution within a prescribed total variation distance~(TVD). We study when such sampling is efficiently classically simulable as a function of system size, the number of magic states, and the noise level. For concreteness in later sections, we instantiate this framework with qubit circuits that inject $|H\rangle$ states subject to dephasing noise and with fermionic circuits that inject a simple four-mode non-Gaussian resource subject to dephasing-like noise or particle loss, but our analysis is formulated for the general setting above.

Within this framework, we aim to identify noise thresholds that separate regimes exhibiting quantum advantage in sampling from regimes that admit efficient classical simulation, under the assumption that noise acts only on magic resources. We do so by constructing and analyzing a classical sampling algorithm tailored to this setting. This section fixes notation and assumptions; the subsequent sections instantiate concrete noise channels and magic resources and develop the corresponding algorithmic and complexity results.

\begin{figure*}
    \centering
    \includegraphics[width=1.0\linewidth]{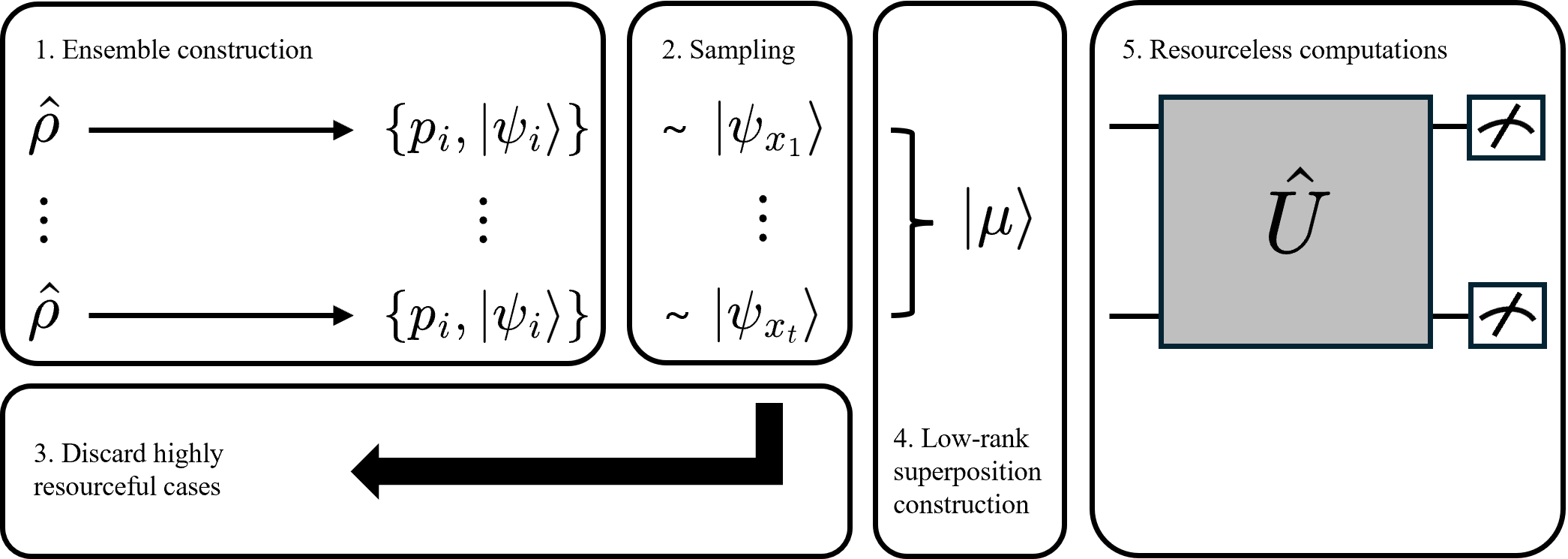}
    \caption{Blueprint of our problem and algorithm. The goal of our strategy is to find a proper low-rank resourceless superposition state. Since the state $|0^{\otimes n}\rangle$ does not affect the resourceless rank of the output state, we present $\hat{\rho}^{\otimes t}$ only in the figure.}
    \label{fig:Total algorithm}
\end{figure*}

\subsection{Algorithm description}\label{Sec:Algorithm}

The key idea of our algorithm is to adopt a pure-state sampling (trajectory) picture of the noisy input.
Since any density operator admits a convex decomposition into pure states, $\hat{\rho}=\sum_i p_i |\psi_i\rangle\langle\psi_i|$, to simulate a resourceless circuit with input $\hat{\rho}$, it suffices to sample $|\psi_i\rangle$ according to $\{p_i\}_i$ and simulate the circuit on that pure input.
Accordingly, our algorithm draws $t$ i.i.d.\ samples from the ensemble $\{p_i,|\psi_i\rangle\}$ and uses their tensor product as the (magic) part of the circuit input (we identify a classical description with the state itself for brevity).
Because all circuit operations are resourceless, the simulation cost is then governed by the resourceless structure of this sampled input: each draw may be resourceless or magic, and the number of sampled magic components controls the complexity, with the worst case when all $t$ samples are magic.
Moreover, since the ancillary factor $|0^{\otimes n}\rangle$ does not affect the relevant resourceless rank, in what follows, we focus on analyzing the noisy magic state.

In the pure-state sampling procedure, when the probability of drawing a large number of magic states is sufficiently small, discarding such rare cases incurs only a small approximation error. 
Consequently, we can reduce the worst-case level via this part by truncating high-resource configurations, namely, samples containing too many magic components. To rigorously analyze this point, we note that the TVD between the probability distributions with and without this truncation method is upper-bounded by the trace distance between the density operators corresponding to the two cases. Therefore, we will focus on the density operator after the truncation method to demonstrate the performance of this step.

Furthermore, even when many magic components are sampled, the computational cost can still be reduced if the resulting pure input state admits a compact representation as a low-rank superposition of resourceless states.
Building on the approach of Ref.~\cite{Bravyi}, which constructs a low-rank stabilizer approximation to $|H^{\otimes t}\rangle$, we generalize this idea and show that the same principle applies in broader settings.
Once such a representation is obtained with small approximation error, the remaining simulation requires evaluating only a few phase-sensitive resourceless components~\cite{Stabilizer_rank,Gaussian_rank2}, leading to substantial computational savings even for samples containing many magic components.

In summary, our algorithm can be outlined as follows:
\begin{enumerate}
    \item Choose a suitable ensemble decomposition of $\hat{\rho}$. It may consist of resourceless and magic states. (See Sec.~\ref{Sec:Ensemble})
    \item Sample $t$ pure states from the ensemble in step~1, and construct a single pure state by the tensor product. If the number of magic components exceeds the truncation threshold, discard the sample and resample. (See Sec.~\ref{Sec:Truncation})
    \item 
    Given the retained sample, construct a low-rank superposition of resourceless states that maintains high fidelity with it. (See Sec.~\ref{Sec:Low-rank})
\end{enumerate}
After obtaining the low-rank superposition of resourceless states, the remaining computation consists of only a few phase-sensitive evaluations of resourceless components, which require computational costs proportional to the resourceless rank of the input state. In other words, each resourceless computation requires just $O(\text{poly}(n))$, because $t = O(\text{poly}(n))$. The overall workflow is illustrated in Fig.~\ref{fig:Total algorithm}.
From now on, we will describe each step in more detail.

\subsubsection{Ensemble conditions}\label{Sec:Ensemble}
Since many different ensembles can represent the same given density operator $\hat{\rho}$, selecting an appropriate ensemble that minimizes the computation cost for the simulation is crucial for the overall algorithm. 
In this work, we choose a particular ensemble $\{p_i, |\psi_i\rangle\}$ corresponding to $\hat{\rho}$ in which each $|\psi_i\rangle$ is either a resourceless state or a magic state of the form
\begin{align}
        |\psi_i\rangle & = \frac{1}{2\nu}(|\psi_i^0\rangle + |\psi_i^1\rangle),
\end{align}
where $|\psi_i^0\rangle$ and $|\psi_i^1\rangle$ are resourceless states such that $\langle\psi_i^0|\psi_i^1\rangle=\langle\psi_i^1|\psi_i^0\rangle = 2\nu^2-1$. We note that a superposition of resourceless states is in general resourceful. A simple example is $|H\rangle = (|0\rangle+|+\rangle)/(2\cos(\pi/8))$ in the qubit system, which is a superposition of two stabilizer states, $|0\rangle$ and $|+\rangle$.

We emphasize that since a given density operator admits multiple valid ensembles, alternative choices may exist, which further reduce the overall simulation complexity beyond that considered here.

\subsubsection{Sampling procedure and truncation method}\label{Sec:Truncation}
Based on the ensemble $\{p_i, |\psi_i\rangle\}$ obtained in the previous subsection, we now describe the associated sampling and truncation procedure to reduce the computational cost for simulating the system. We emphasize that since our objective is to approximate the process of obtaining a measurement outcome from a quantum circuit, rather than to compute its probability distribution, it suffices to sample pure states from the ensemble prior to applying the quantum circuit and performing the measurement in the trajectory manner, instead of directly simulating the full density operator $\hat{\rho}$.


More precisely, we first sample $t$ pure states $|\psi_{x_1}\rangle,\dots,|\psi_{x_t}\rangle$ independently from the ensemble according to the probabilities $\{p_i\}$, where $x_j$'s are indices for the states in the given ensemble. Next, we construct $|\Psi\rangle = \bigotimes_{j=1}^t|\psi_{x_j}\rangle$ and execute phase-sensitive resourceless computations with $|0^{\otimes n}\rangle\langle0^{\otimes n}|\otimes |\Psi\rangle\langle\Psi|$ instead of $|0^{\otimes n}\rangle\langle 0^{\otimes n}|\otimes \hat{\rho}^{\otimes t}$. Due to the probabilistic structure of the density operator, the probability of obtaining a measurement outcome from this procedure is equivalent to the probability from the original circuit.


Since the resourceless rank of the constructed pure state increases as the number of sampled magic states, the time complexity of simulating the quantum circuit with sampled states depends on how many magic states are drawn. The worst case occurs when all $t$ sampled states are magic states, yielding a superposition of $2^t$ resourceless states. In this case, classical simulation of the worst case requires $2^t$ resourceless computations according to the conditions in Sec.~\ref{Sec:Ensemble}. 
However, if the probability of such highly resourceful situations is sufficiently small, neglecting such cases may induce only a small approximation error in the sampling task.

More precisely, let $|\psi_i\rangle$ be a resourceless state for $i = 1,\dots,K$ and a magic state for $i = K+1,\dots,L$, without loss of generality. Then, the number of magic states in a sampled configuration follows a binomial distribution with success probability
\begin{align}
p = \sum_{i=K+1}^{L} p_i,
\end{align}
and failure probability $1-p = \sum_{i=1}^{K} p_i$.
Therefore, the probability of sampling at least $k+1$ magic states from the ensemble $\{p_i, |\psi_i\rangle\}$ can be evaluated using the upper tail of this binomial distribution.



\begin{lemma}\label{Truncation}
    Let $\hat{\rho}$ be a density matrix described as an ensemble $\{p_i, |\psi_i\rangle\}$ and $\delta_1>0$ be given. Let $k$ be the smallest integer such that $k\geq tp$ and
    \begin{align}\label{Eq:Truncation condition}
        D\left(\frac{k+1}{t}\bigg\| p\right) & \geq \frac{1}{t}\log\left(\frac{1}{\delta_1}\right),
    \end{align}
    where $D(a\| p)$ is the Kullback-Leibler divergence of a binomial distribution, defined as
    \begin{align}
    D(a\| p) & \equiv a\ln\frac{a}{p} + (1-a)\ln\frac{1-a}{1-p}.
    \end{align}
    Consider a truncation procedure described as 
    \begin{enumerate}
        \item Sample $t$ states $|\psi_{x_j}\rangle \sim \{p_i,|\psi_i\rangle\}$,
        \item If the number of sampled magic states is larger than $k$, then we discard the case. If not, construct $|\Psi\rangle = \bigotimes_{j=1}^t|\psi_{x_j}\rangle$.
    \end{enumerate}
    Then, the trace distance between $\hat{\rho}^{\otimes t}$ and the density operator generated by this procedure is at most $\delta_1$.
\end{lemma}
\begin{proof}
    With the above setting of the ensemble, the probability of sampling at least $k+1$ magic states during $t$ rounds has an upper bound,
    \begin{align}\label{Upper bound}
            \sum_{j=k+1}^t \begin{pmatrix}
                t \\ j
            \end{pmatrix}p^j(1-p)^{t-j} & \leq \exp\left[-tD\left(\frac{k+1}{t}\bigg\| p\right)\right],
    \end{align}
    by the Chernoff bound. By the assumption, it implies that this upper bound is upper-bounded by $\delta_1$. Now, denote $\hat{\rho} = \hat{\sigma}_0 + \hat{\sigma}_1$, where $\hat{\sigma}_0$ consists of states including at most $k$ magic states, and $\hat{\sigma}_1$ consists of states including at least $k+1$ magic states. Then, our desired trace distance is
    \begin{align}
            \frac{1}{2}\left\|\hat{\rho}^{\otimes t}-\frac{\hat{\sigma}_0}{\textrm{Tr}[\hat{\sigma}_0]}\right\|_1 & \leq \frac{1}{2}\left(\left\|\hat{\rho}^{\otimes t} - \hat{\sigma}_0\right\|_1 + \left\|\hat{\sigma}_0 - \frac{\hat{\sigma}_0}{\textrm{Tr}[\hat{\sigma}_0]}\right\|_1\right) \\
            & \leq \textrm{Tr}[\hat{\sigma}_1] \\
            & \leq \delta_1,
    \end{align}
    where the second line is derived from $1-\textrm{Tr}[\hat{\sigma}_0] = \textrm{Tr}[\hat{\sigma}_1]$, and the last line is derived from the fact that $\textrm{Tr}[\hat{\sigma}_1]$ is equivalent to the probability of sampling at least $k+1$ states during $t$ rounds, described in Eq.~\eqref{Upper bound}.
\end{proof}

To discuss its performance, we remark that the probability of sampling the truncated cases is upper-bounded by $\delta_1$. Hence, after $O((1-\delta_1)^{-1})$ trials on average, we can obtain a non-truncated case. Therefore, we can conclude that its time complexity is $O((1-\delta_1)^{-1})$ on average.

After the truncation method, the worst case is modified to the case when we sample $k$ magic states from $\{p_i,|\psi_i\rangle\}$, where $k$ is defined in Lem.~\ref{Truncation}. Similarly, a constructed pure state in the modified worst case is a superposition of $2^k$ resourceless states. Therefore, the required number of resourceless computations for the worst case is reduced, compared to the original worst case, which requires $2^t$ resourceless computations.

Lastly, consider the case when $\delta_1$ is exponentially small, i.e. $\delta_1 = O(1/\exp(n))$. Intuitively, the smaller $\delta_1$ is, the larger the number of magic states in the modified worst case $k$ is. Following this intuition, one can expect that the truncation method implies only a little computational enhancement in this extreme case. More formally, $\delta_1 = O(1/\exp(n))$ implies that the right-hand side of Eq.~\eqref{Eq:Truncation condition} becomes $\Omega(\text{poly}(n)/t)$, which makes $k$ is almost close or equal to $t$. Hence, when $\delta_1$ is exponentially small, simulation with the truncation method is almost as close to simulating the original noisy circuit.

\subsubsection{Finding a low-rank superposition of resourceless states}\label{Sec:Low-rank}
Even after discarding highly resourceful cases, the constructed pure state typically remains resourceful. For example, in the truncated worst case where $k$ magic states are sampled during the construction step, the state becomes a superposition of $2^k$ resourceless states, which may still require substantial computational resources.
In this subsection, we further reduce this cost by finding a superposition consisting of only a few resourceless states that preserves high fidelity with the original state, following and generalizing the approach introduced in Ref.~\cite{Bravyi}. For generality, let us assume that $m$ magic states are sampled, where $m \leq k$.

\begin{lemma}\label{Low rank}
    Let $\delta_2>0$ and $|\psi_1\rangle, |\psi_2\rangle,\dots, |\psi_m\rangle$ be magic states of the form
    \begin{align}
            |\psi_i\rangle & = \frac{1}{2\nu}\left(|\psi_i^0\rangle + |\psi_i^1\rangle\right),
    \end{align}
    where both $|\psi_i^0\rangle$ and $|\psi_i^1\rangle$ are resourceless states satisfying $\langle\psi_i^0|\psi_i^1\rangle = \langle\psi_i^1|\psi_i^0\rangle = 2\nu^2-1$ for all $i=1,2,\dots,m$, so that $\nu$ is real. Denote $|\Psi\rangle = \bigotimes_{i=1}^m |\psi_i\rangle$. Then, there exists a state $|\mu\rangle$ such that
    \begin{align}
            \left|\langle\Psi|\mu\rangle \right|^2 & \geq 1-\delta_2,
    \end{align}
    and $|\mu\rangle$ can be expressed as a superposition of $2^l$ resourceless states satisfying $2^l\leq 4\nu^{-2m}\delta_2^{-1}$.
\end{lemma}

We give a sketch of the proof here to illustrate this subroutine. The main idea is to work on $\mathbb{Z}_2^m$ rather than $|\Psi\rangle$. Because $|\psi_1\rangle,\dots,|\psi_m\rangle$ share the same structure described in Sec.~\ref{Sec:Ensemble}, the state $|\Psi\rangle$ can be represented as a superposition of pure states naturally indexed by bit strings in $\mathbb{Z}_2^m$. For any linear subspace $\mathcal{L} \subset \mathbb{Z}_2^m$ of dimension $l$, we associate a corresponding normalized state $|\mathcal{L}\rangle$.
One can check that the fidelity between $|\Psi\rangle$ and $|\mathcal{L}\rangle$ depends only on $l,\nu$, and its normalization factor $Z(\mathcal{L})$. Hence, our goal reduces to finding a suitable subspace $\mathcal{L}^*$ for which this fidelity is large. Using Markov's inequality, by randomly choosing $O(\delta_2^{-1})$ subspaces, we can obtain a suitable $\mathcal{L}^\ast$ with constant probability. A detailed proof is provided in App.~\ref{App:Proof of Low rank Thm} for completeness.

Let us discuss its performance. For the case where we sample $m$ magic states, this sub-algorithm is summarized to finding a suitable subspace $\mathcal{L}\subset\mathbb{Z}_2^m$ of dimension $l$. We can compute $Z(\mathcal{L})$ in $O(2^l)$ time. For finding $\mathcal{L}^*$, by randomly choosing $O(\delta_2^{-1})$ subspaces $\mathcal{L}$, we can obtain $\mathcal{L}^*$ with constant probability, by Markov's inequality. Therefore, the total time complexity of finding $\mathcal{L}^*$ is $O(2^l\delta_2^{-1}) = O(\nu^{-2m}\delta_2^{-2})$. Hence, the modified worst case, when we sample $k$ magic states, requires $O(\nu^{-2k}\delta_2^{-2})$ time. We again remark that this discussion is based on Ref.~\cite{Bravyi}.

In short, we adopt an approximate classical algorithm for finding a superposition of at most $O(\nu^{-2k}\delta_2^{-1})$ resourceless states, preserving high fidelity with $|\Psi\rangle$, in at most $O(\nu^{-2k}\delta^{-2}_2)$ time, with small approximation error.

While this subroutine is useful when many magic states are sampled, it may not provide any advantage when few magic states are sampled. More concretely, consider the case in which the number of sampled magic states $m$ satisfies $2^m \leq 4\nu^{-2m}\delta_2^{-1}$. In this case, this subroutine cannot reduce the rank of the constructed state. We therefore introduce a threshold $m_0$ as the largest integer such that $2^{m_0} \leq 4\nu^{-2m_0}\delta_2^{-1}$ and apply the low-rank approximation subroutine only when $m > m_0$.

Let us discuss the case when $\delta_2$ is exponentially small. As with the truncation method, if $\delta_2$ is exponentially small, finding a low-rank resourceless superposition state part yields a state close to the original high-rank resourceless superposition state. More formally, $4\nu^{-2m_0}\delta_2^{-1} = \Omega(\exp(n))$ if $\delta_2 = O(1/\exp(n))$, so that $m_0$ becomes almost close or equal to $k$. It means that the rank of the result of this part is almost close to the rank of the state in the modified worst case via the truncation method. Furthermore, an extremely small $\delta_2$ yields the case when $m_0 = k$. We again mention that the subroutine described in this subsection does not work in this extreme case.

Over the whole of our strategy, we introduce two error parameters: $\delta_1$ and $\delta_2$, which are described in Sec.~\ref{Sec:Truncation} and Sec.~\ref{Sec:Low-rank} respectively. Since these error parameters are independent, one can optimize them to ensure that the total TVD is tightly close to the desired value. However, the complex form of the total TVD prevents us from analytical optimization. We remark that the subroutine described in Sec.~\ref{Sec:Low-rank} has a threshold. It makes an upper bound of the total TVD a complex form:
\begin{align}\label{Eq:Total TVD}
    \Delta(\delta_1, \delta_2) & \coloneq \delta_1 + \frac{\textrm{Pr}[m_0< X\leq k]}{\textrm{Pr}[X\leq k]}\sqrt{\delta_2},
\end{align}
where $\textrm{Pr}[X = m]$ is the probability of sampling $m$ magic states from $\{p_i, |\psi_i\rangle\}$. 
In this paper, for simplicity, we set error parameters as a simple case for the analysis of our procedure: $\delta_1=\delta/2$ and $\delta_2 = \delta^2/4$, which induces TVD smaller than $\delta$, while a more optimized choice may be found. 

Finally, the total time complexity for generating a proper low-rank resourceless superposition state is $O((1-\delta/2)^{-1}+ \nu^{-2k}\delta^{-4})$. However, our setting yields a simpler time-complexity analysis for our algorithm. More precisely, since $\delta$ is an upper bound of TVD, it should satisfy $\delta\leq 1$. It implies that the time complexity of the truncation method, $O((1-\delta/2)^{-1})$, becomes a constant scale. Hence, the total time complexity can be simplified to $O(\nu^{-2k}\delta^{-4})$ in this setting if the part described in Sec.~\ref{Sec:Low-rank} works.

Let us discuss the case where $\delta$ is exponentially small. As we mentioned earlier, this case, making both $\delta_1$ and $\delta_2$ exponentially small, implies that both subroutines of our algorithm yield only a little computational enhancement. Furthermore, an extremely small $\delta$ can yield the case when $2^k\leq 16\nu^{-2k}\delta^{-2}$. In this extreme case, we cannot reduce the resourceless rank of the state in the modified worst case. Therefore, the rank of the result of our algorithm is $2^k$, which is almost close or equal to the rank of the state in the original worst case. In other words, our simulation strategy is close to simulating the original noisy circuit when $\delta$ is extremely small.

    Therefore, for convenience, we discuss the case where $\delta = \Omega(1/\text{poly}(n))$ and $2^k > 16\nu^{-2k}\delta^{-2}$ in the later discussions. Lastly, we again remark that the remaining work is just a few phase-sensitive evaluations of resourceless components, which requires at most $O(\nu^{-2k}\delta^{-2}\text{poly}(n))$ time. Therefore, the total computational cost for simulating a noisy quantum circuit is at most $O(\nu^{-2k}\delta^{-2}(\text{poly}(n)+\delta^{-2}))$ time. However, it can be simplified as $O(\nu^{-2k}\text{poly}(n))$ under our setting, because $\delta = \Omega(1/\text{poly}(n))$.

In short, our strategy described in Sec.~\ref{Sec:Algorithm} can be summarized as follows:
\begin{theorem}
    Consider a quantum circuit that injects $t$ noisy magic states $\hat{\rho}^{\otimes t}$. Let $\delta = \Omega(1/\textup{poly}(n))$ be given. Then, there is a classical algorithm for approximately simulating this circuit in at most $O(\nu^{-2k}\textup{poly}(n))$ time, up to $\delta$ TVD.
\end{theorem}


\section{Qubit system}\label{Sec:Qubit}
In Sec.~\ref{Sec:General framework}, we introduced the general problem setting and outlined our overall strategy for classical simulation. We now apply this framework to concrete scenarios. As a first application, we focus on the qubit system, one of the primary systems for quantum computation. In this section, we analyze qubit circuits composed of noisy magic states, Clifford gates, and computational-basis measurements, and demonstrate how the proposed algorithm can be explicitly implemented in this setting.

\subsection{Problem description}\label{Sec:Qubit_problem}
As we discussed previously, the Gottesman–Knill theorem states that circuits composed solely of stabilizer states, Clifford gates, and computational-basis measurements are classically simulable and thus cannot exhibit quantum advantage in the qubit system~\cite{Gottesman, Stabilizer_tableau}. Accordingly, the introduction of magic components is essential for achieving quantum advantage, and their inclusion enables universal quantum computation~\cite{Nielsen_Chuang}. A representative approach is magic-state injection, such as $T$-gadgetization~\cite{Gadget}: once suitable magics states are available, one can obtain quantum advantage and universality even when all remaining operations are restricted to Clifford gates.

In practice, however, quantum systems are inevitably subject to physical noises~\cite{Classical_RCS1, Classical_RCS2, Classical_RCS3, Noisy_sampling_RCS, RCS_experiment}. In particular, preparing large numbers of high-fidelity magic states remains one of the major experimental challenges~\cite{SPAM1, SPAM2, Magic_distillation1, Magic_distillation2}. Such imperfections can shrink, or even close, the gap between quantum and classical computation. In this section, we therefore investigate how noise acting on magic states affects the resulting computational power. For clarity, we focus on dephasing noise represented by the Kraus operators
\begin{align}
        \hat{K}_0 = \sqrt{1-p}\hat{I},\quad
        \hat{K}_1 = \sqrt{p}\hat{Z},
\end{align}
where $p \in [0,1/2]$ is the dephasing rate.

Consider a magic state $|H\rangle = \cos(\pi/8)|0\rangle+\sin(\pi/8)|1\rangle$, which is Clifford equivalent to $|T\rangle$. More precisely, one can easily check that
\begin{align}
    |H\rangle & = e^{-i\frac{\pi}{8}}\hat{S}\hat{H}|T\rangle,
\end{align} 
where $\hat{H}$ is the Hadamard gate and $\hat{S}$ is the phase gate, so that both are Clifford gates. The output state after the dephasing channel on $|H\rangle$ is written in the computational basis as follows:
\begin{align}
        \hat{\rho}_{\text{qubit}} & = \begin{pmatrix}
            \cos^2(\frac{\pi}{8}) & (1-2p)\cos(\frac{\pi}{8})\sin(\frac{\pi}{8}) \\
            (1-2p)\cos(\frac{\pi}{8})\sin(\frac{\pi}{8}) & \sin^2(\frac{\pi}{8})
        \end{pmatrix},
\end{align}
which describes the effect of the dephasing noise on $|H\rangle$.
Hence, we now present a way to replace the tensor product of such noisy magic states with a low-rank superposition of resourceless states in simulating the quantum circuit.


\subsection{Classical simulation}
As discussed in Ref.~\cite{Bravyi}, the state $|H\rangle$ can be written as the superposition of two stabilizer states,
\begin{align}
        |H\rangle & = \frac{1}{2\nu_{\text{qubit}}}(|0\rangle+|+\rangle),
\end{align}
where $\nu_{\text{qubit}} \equiv \cos(\pi/8)$. Consequently, Lem.~\ref{Low rank} can be applied to the noiseless case as in Ref.~\cite{Bravyi}. Our goal here is to extend this picture to noisy magic states. Intuitively, as the dephasing rate increases, we expect the computational power associated with $|H\rangle$ to decrease. In particular, the off-diagonal terms of $\hat{\rho}_{\text{qubit}}$ vanish at $p=1/2$, at which point $\hat{\rho}_{\text{qubit}}$ can be described as a mixture of stabilizer states, $\{|0\rangle,|1\rangle\}$. To interpolate smoothly between the noiseless and fully dephased regimes, we construct and analyze $p$-dependent ensembles for $\hat{\rho}_{\text{qubit}}$.

Our strategy for identifying suitable ensembles under this observation is as follows.
We first determine the threshold beyond which $|H\rangle$ is no longer needed in any ensemble representation of $\hat{\rho}_{\text{qubit}}$. 
Since the noiseless case is entirely described by $|H\rangle$, we consider ensembles that represent $\hat{\rho}_{\text{qubit}}$ using $|H\rangle$ together with single-qubit stabilizer states.
Since there are only six single-qubit stabilizer states, this analysis can be performed analytically. 
Furthermore, we formulate and solve an optimization problem to minimize the probability associated with $|H\rangle$. The detailed derivation is provided in App.~\ref{App:Optimization}. In other words, we construct an ensemble tailored to our algorithm by minimizing the probability of sampling non-stabilizer states. 
From this optimization, we obtain a critical value of the dephasing parameter: $\hat{\rho}_{\text{qubit}}$ admits a representation consisting solely of stabilizer states if and only if $p\geq (1-\tan(\pi/8))/2 \approx 0.2929$.  This analysis can also be extended to the complementary regime $p \in [1/2, 1]$, as $\hat{\rho}_{\text{qubit}}$ exhibits symmetry about the axis $p = 1/2$. In this region, the corresponding boundary is $p=(1+\tan(\pi/8))/2 \approx 0.7071$. 
In what follows, we focus on the interval $p \in [0,1/2]$, since the behavior for $p \in [1/2,1]$ can be obtained by symmetry.

When the noise rate lies above the relevant threshold, the noisy magic state $\hat{\rho}_{\text{qubit}}$ can thus be expressed as an ensemble consisting solely of stabilizer states, and our simulation reduces to stabilizer-based classical algorithms. For completeness, we now specify our proposed ensembles for the different parameter regimes. First, for $p \in [0,(1-\tan(\pi/8))/2]$, we take the ensemble to be
\renewcommand{\arraystretch}{1.5} 

\begin{center}
\begin{tabular}{@{\hspace{1.4cm}}c@{\hspace{1.4cm}\vrule\hspace{1.4cm}}c@{\hspace{1.4cm}}}
  \textbf{State} & \textbf{Probability} \\ \hline
  $|0\rangle$ & $(1+\sqrt{2})p$ \\
  $|+\rangle$ & $p$ \\
  $|H\rangle$ & $1-(2+\sqrt{2})p$
\end{tabular}
\end{center}

which consists of two stabilizer states and one magic state. Consequently, on $p\in[(1-\tan(\pi/8))/2, 1/2]$, our proposed ensemble is

\begin{center}
\begin{tabular}{@{\hspace{1.4cm}}c@{\hspace{1.4cm}\vrule\hspace{1.46cm}}c@{\hspace{1.46cm}}}
  \textbf{State} & \textbf{Probability} \\ \hline
  $|0\rangle$ & $\frac{1}{2}+\frac{1}{\sqrt{2}}p$           \\
  $|1\rangle$ & $\frac{1-\sqrt{2}}{2} + \frac{1}{\sqrt{2}}p$                         \\
  $| +\rangle$& $\frac{1}{\sqrt{2}}-\sqrt{2}p$
\end{tabular}
\end{center}

As desired, one can easily check that our ensembles consist of only stabilizer states on $p\in[(1-\tan(\pi/8))/2, 1/2]$, so that every sampled state becomes a stabilizer state.

Next, let us discuss how to simulate a sampled state. As mentioned above, the region near the fully dephased point, $p\in[(1-\tan(\pi/8))/2,1/2]$, is computationally trivial since every sampled state is a stabilizer state. Hence, we focus on the region $p\in[0,(1-\tan(\pi/8))/2]$. To apply the truncation method, we first analyze the probability of sampling many non-stabilizer states. In the first ensemble, the probability of sampling a non-stabilizer state in a single draw is $1-(2+\sqrt{2})p$. By Lem.~\ref{Truncation}, we define $f_{\text{qubit}}(t,p,\delta)$ as the minimum number of sampled non-stabilizer states satisfying $f_{\text{qubit}}(t,p,\delta)\geq t(1-(2+\sqrt{2})p)$ and
\begin{align}\label{f0}
        D\left(\frac{f_{\text{qubit}}(t,p,\delta)+1}{t}\bigg\| 1-(2+\sqrt{2})p\right) & \geq \frac{1}{t}\log\left(\frac{2}{\delta}\right).
\end{align}
Finally, suppose that we sample $m$ non-stabilizer states $|H\rangle$ and $m$ satisfies $m > 16\nu_{\text{qubit}}^{-2m}\delta^{-2}$. We remark that if $m$ does not satisfy this condition, i.e., the number of sampled magic states is sufficiently small, then the next step is not applied. As discussed in Ref.~\cite{Bravyi}, we can find a low-rank stabilizer superposition state that preserves high fidelity with $|H\rangle^{\otimes m}$, since $|H\rangle$ can be decomposed into two stabilizer states $|0\rangle$ and $|+\rangle$. It can also be derived from Lem.~\ref{Low rank}.
\begin{lemma}
    There is a state $|\mu\rangle$ such that
    \begin{equation}
        \begin{split}
            \left|\langle H^{\otimes m}|\mu\rangle\right|^2 & \geq 1-\frac{\delta^2}{4}
        \end{split}
    \end{equation}
    and $|\mu\rangle$ can be expressed as a superposition of $2^l$ stabilizer states satisfying $2^l\leq16\nu_{\textup{qubit}}^{-2m}\delta^{-2}$, where $\nu_{\textup{qubit}} = \cos(\pi/8)$.
\end{lemma}
The worst case occurs when the number of sampled $|H\rangle$ is $f_{\text{qubit}}(t,p,\delta)$. Therefore, on $p \in [0, (1-\tan(\pi/8))/2]$, any sampled state from $\hat{\rho}_{\text{qubit}}^{\otimes t}$ can be approximated by a pure state with at most $O(\nu_{\text{qubit}}^{-2f_{\text{qubit}}(t,p,\delta)}\delta^{-2})$ stabilizer rank, and we can find it in at most $O(\nu_{\text{qubit}}^{-2f_{\text{qubit}}(t,p,\delta)}\delta^{-4})$ time. We again remark that the remaining work is just a few phase-sensitive Clifford computations~\cite{Stabilizer_rank}. Combining the case when $p\in [(1-\tan(\pi/8))/2, 1/2]$, application of our strategy can be summarized as follows:
\begin{theorem}
    Assume $p\in [0, (1-\tan(\pi/8))/2]$ and $\delta = \Omega(1/\textup{poly}(n))$ is given. There is a classical algorithm for approximately simulating the circuit in at most $O(\nu_{\textup{qubit}}^{-2f_{\textup{qubit}}(t,p,\delta)}\textup{poly}(n))$ time, up to $\delta$ TVD. Furthermore, on $p\in [(1-\tan(\pi/8))/2, 1/2]$, there is a classical algorithm for exactly simulating the circuit in $O(\textup{poly}(n))$ time.
\end{theorem}

In the end, we investigate when this noisy scheme collapses into a classical computational regime. Since we have two independent variables, the noise rate $p$ and the number of magic states $t$, we answer the following two questions: Which condition of $p$ (or $t$) makes this noisy scheme collapse into a classical simulable regime for given $t$ (or $p$)? Due to our algorithm inducing detailed time complexity, we can investigate a boundary where the time complexity becomes a polynomial scale, which is a classical simulable regime. Thm.~\ref{Thm:Qubit_poly1} answers these questions.
\begin{theorem}\label{Thm:Qubit_poly1}
    Consider the ensemble defined on $p \in [0, (1-\tan(\pi/8))/2]$ and let $\delta = \Omega(1/\textup{poly}(n))$. Assume that the number of injected magic states $t$ is given. If the noise rate $p$ satisfies
    \begin{align}
        p & = \frac{1-\tan\left(\frac{\pi}{8}\right)}{2}\left(1-O\left(\frac{\log n}{t}\right)\right),
    \end{align}
    then the time complexity becomes $O(\textup{poly}(n))$. Conversely, assume that the noise rate $p$ is given. If the number of magic states $t$ satisfies
    \begin{align}
        t & = O\left(\frac{\log n}{1-(2+\sqrt{2})p}\right),
    \end{align}
    then the time complexity becomes $O(\textup{poly}(n))$.
\end{theorem}

Detailed proofs are provided in App.~\ref{App:Qubit_Time complexity}. These results align with intuition. In Thm.~\ref{Thm:Qubit_poly1}, the threshold on $p$ increases with $t$, indicating that the more magic states are injected, the more tolerant the computational power becomes to physical noise. The extreme cases are also consistent. When $t = 0$, i.e., no magic states are injected, the classically simulable region extends over all $p$, so the circuit is classically simulable for any noise rate, exactly as guaranteed by the Gottesman–Knill theorem.

The opposite case exhibits a similar behavior. As desired, the threshold on $t$ grows as the noise rate approaches the critical value $(1-\tan(\pi/8))/2$. In the opposite extreme, when $p \to 0$, the number of injected magic states that remains classically simulable scales as $O(\log n)$. This again matches the Gottesman–Knill picture, since injecting $O(\log n)$ copies of $|H\rangle$ yields only polynomial stabilizer rank. Moreover, when $\delta$ is taken to be extremely small, our algorithm approximates the original noisy circuit very closely, as one would expect.

Lastly, since we do not apply noisy channels to Clifford operations, it is natural to compare our results with MSD, which produces higher-fidelity magic states from many noisy ones under ideal Clifford processing~\cite{Magic_distillation1, Magic_distillation2, Magic_distillation3}. For $|H\rangle$, the known boundary between distillable and non-distillable regimes is tight; in our situation, this boundary is $ p = (1-\tan(8/\pi))/2$, which coincides with the boundary of our proposed ensemble descriptions~\cite{Magic_distillation3}. Furthermore, a noisy $|H\rangle$ can be expressed as a convex combination of only stabilizer states when it is beyond this boundary, which coincides with our proposed ensemble~\cite{Magic_distillation1}.

However, distillability from an informational perspective alone does not guarantee that every distillable noisy $|H\rangle$ supports universal quantum computation at a practical level, which restricts the depth of a quantum circuit to a polynomial scale, i.e, computational cost. Intuitively, if a noisy magic state lies very close to the distillation threshold, then polynomially many copies may be insufficient to distill polynomially many high-fidelity magic states. In other words, when we inject highly noisy magic states near this boundary, a polynomial number of injections may fail to provide enough computational power.

Our result formalizes this intuition. Consider a noisy magic state $\hat{\rho}$ in the classically simulable regime, but still in the distillable regime according to our analysis (Inside the blue box but $p \in [0, (1-\tan(\pi/8))/2]$ in Fig.~\ref{fig:Total regime}). 
As we mentioned, since it is distillable, one can obtain polynomially many high-quality $|H\rangle$, which induce universal quantum computation with a depth of a polynomial scale, from sufficiently many copies of $\rho$. Now, assume that one can distillate sufficiently many high-quality $|H\rangle$, which grants classical intractability to a quantum circuit, from polynomially many $\rho$ and distillation protocols. Our result directly implies that every quantum circuit capable of universal quantum computation with a depth of a polynomial scale is classically simulable. It contradicts a firm belief in the computational power of a quantum computer. Therefore, polynomially many $\hat{\rho}$ cannot provide sufficient computational power to a quantum circuit, even if they are distillable.

In summary, we analytically characterize how dephasing noise degrades the computational power of qubit circuits that would otherwise support universal quantum computation. Additionally, we identify noise and resource thresholds in terms of $p$ and $t$ at which the noisy scheme collapses into a classically simulable regime. Finally, we verify that our analysis is consistent with the noiseless limit described by the Gottesman–Knill theorem.

\section{Fermionic system}\label{Sec:Fermion}
Following the analysis of the qubit system in Sec.~\ref{Sec:Qubit}, we now apply our general simulation framework to the fermionic system, another setting that supports the universal quantum computation. While many fermionic circuits are classically simulable~\cite{Classical_FLO1, Classical_FLO2, Gaussian_rank1, Gaussian_rank2}, the use of non-Gaussian resources enables universal quantum computation~\cite{Fermion_computing, Fermion_computing2}. Furthermore, it is revealed that every pure non-Gaussian state possesses sufficient computational power, in terms of universal quantum computation, even if we adopt only Gaussian operations~\cite{Magic}. However, generating non-Gaussian states requires non-Gaussian operators, which have complicated structures, such as higher-order terms of Majorana operators. It induces the hardness of implementing high-fidelity non-Gaussian states. In this section, we describe a fermionic circuit composed of noisy magic states, Gaussian gates, and number-basis measurements, and demonstrate how our algorithm performs in this noisy setting.

\subsection{Problem description}\label{Sec:Fermion_problem}

Every pure non-Gaussian state can play the role of a magic state, which enables us to achieve universal quantum computation with adaptive measurements~\cite{Magic}. One of the representative magic states in the fermionic system is $|\psi_4\rangle = (|0011\rangle+|1100\rangle)/\sqrt{2}$, which is one of the representative magic states in the fermionic system~\cite{Fermion_computing2, Fermion_sampling, Fermion_sampling2, Fermion_error1, Fermion_error2}. Here, we inject $|\psi_4^{\otimes t}\rangle$, but suffer from physical noise, into a fermionic circuit consisting of only Gaussian components.
\begin{figure}[t!]
    \centering
    \includegraphics[width=1.0\linewidth]{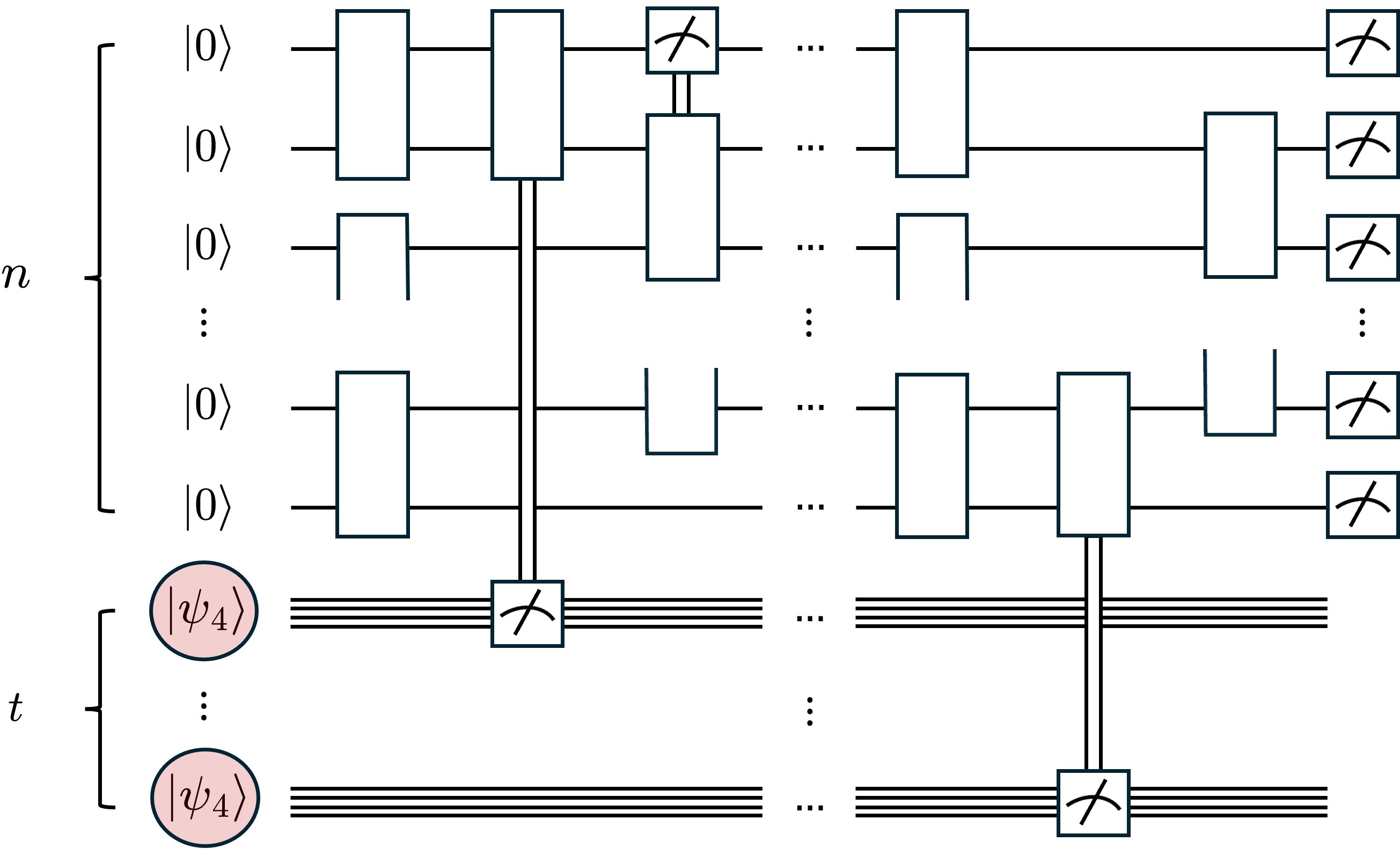}
    \caption{Our desired circuit in the fermionic system. We injects $t$ noisy $4$-mode magic states, $|\psi_4\rangle = (|0011\rangle+|1100\rangle)/\sqrt{2}$. The white boxes represent Gaussian gates. As in the qubit case, we allow adaptive measurements in the injected circuit.}
    \label{fig:Fermion}
\end{figure}

Although the computational power of ideal magic states has already been demonstrated, the computation in a real physical environment can also be affected by various sources of noise. Let us focus on two different noise models: particle loss and dephasing noise. Particle loss can be modeled as a beam-splitter interaction with the environment in vacuum. For instance, the single fermion state transforms as
\begin{align}\label{Particle loss}
        \textrm{Tr}_2\left[\hat{U}|10\rangle\langle10|\hat{U}^\dagger\right] & = \cos^2\lambda|1\rangle\langle1|+\sin^2\lambda|0\rangle\langle0|,
\end{align}
where $\hat{U} = \exp[-i\lambda(\hat{a}_1^\dagger\hat{a}_2 + \hat{a}_2^\dagger\hat{a}_1)]$, $\hat{a}^\dagger_j$ and $\hat{a}_j$ are the creation and annihilation operator on mode $j$. Dephasing noise is also common in real environments because the system is sensitive to phase. The same Kraus operators also describe the dephasing noise on the fermionic system as in the qubit case.

In this section, we analyze the situation of injecting $t$ noisy $|\psi_4\rangle$, suffering particle loss and dephasing noise, into a Gaussian circuit. We again emphasize that this circuit can achieve universal quantum computation in a noiseless environment. To analyze this situation, we start by applying the explicit representation of the impact of noises on $|\psi_4\rangle$ directly. First, when we apply the particle loss to $|\psi_4\rangle$, we can obtain a density operator $\hat{\rho}_{\text{loss}}$ expressed as the following mixed state

\begin{center}
\begin{tabular}{@{\hspace{1.4cm}}c@{\hspace{1.4cm}\vrule\hspace{1.4cm}}c@{\hspace{1.4cm}}}
  \textbf{State} & \textbf{Probability} \\ \hline
  $|\phi_+\rangle$ & $\frac{N}{2}\ $           \\
  $|\phi_-\rangle$ & $\frac{N}{2}\ $ \\
  $|0001\rangle$ & $\frac{1}{2}p(1-p)\ $                         \\
  $|0010\rangle$& $\frac{1}{2}p(1-p)\ $ \\
  $|0100\rangle$ & $\frac{1}{2}p(1-p)\ $ \\
  $|1000\rangle $ & $\frac{1}{2}p(1-p)$,
\end{tabular}
\end{center}

where $|\phi_\pm\rangle \equiv [(1-p)|0000\rangle\pm p|\psi_4\rangle]/\sqrt{N}$ and $N \equiv p^2+(1-p)^2 = 2p^2-2p+1$. One can check that it comprises two non-Gaussian states and four Gaussian states. Similarly, when we apply the dephasing noise on $|\psi
_4\rangle$, we have a density operator $\hat{\rho}_{\text{dep}}$ expressed as the following mixed state
\begin{center}
\begin{tabular}{@{\hspace{0.26cm}}c@{\hspace{0.26cm}\vrule\hspace{1.42cm}}c@{\hspace{1.42cm}}}
  \textbf{State} & \textbf{Probability} \\ \hline
  $C_+|0011\rangle+C_-|1100\rangle$ & $\frac{1}{2}\ $          \\
  $C_-|0011\rangle+C_+|1100\rangle$ & $\frac{1}{2}$, 
\end{tabular}
\end{center}
where 
\begin{align}
        C_\pm & \equiv \frac{1}{2}\left(\sqrt{1+(1-2p)^4}\pm\sqrt{1-(1-2p)^4}\right). 
\end{align}
In this case, the ensemble contains no Gaussian state.

\subsection{Classical simulation}\label{Sec:Fermion_simulation}

First, let us discuss the particle loss case. To utilize our algorithm, we show that $|\phi_\pm\rangle$ can be decomposed into two superpositions of Gaussian states, respectively. We represent $|\phi_{\pm}\rangle$ as the superposition of two states
\begin{align}
        |\phi_\pm\rangle &  = \sqrt{M}\left(|\phi^0_\pm\rangle + |\phi^1_\pm\rangle\right),
\end{align}
where $M \equiv \frac{1}{4N}[2p^2 + (1-p)^2] = \frac{3p^2-2p+1}{4(2p^2-2p+1)}$ and,
\begin{align}
        |\phi^0_\pm\rangle & \equiv |00\rangle\left(\frac{1-p}{2\sqrt{NM}}|00\rangle \pm \frac{p}{\sqrt{2NM}}|11\rangle\right), \\
        |\phi^1_\pm\rangle & \equiv \left(\frac{1-p}{2\sqrt{NM}}|00\rangle \pm\frac{p}{\sqrt{2NM}}|11\rangle\right)|00\rangle.
\end{align}
We emphasize that both states are expressed as tensor products of two-mode states. Since every two-mode state is Gaussian~\cite{Gaussian, Magic}, both states are tensor products of Gaussian states, which implies that both states are also Gaussian. A simple proof of the Gaussianity of a two-mode even state is provided in App.~\ref{Two mode}. Next, we investigate the probability of sampling a lot of non-Gaussian states to apply the truncation method. In our ensemble, the probability of sampling a non-Gaussian state once is $N = 2p^2-2p+1$. By following Lem.~\ref{Truncation}, we define $f_{\text{loss}}(t,p,\delta)$ as the minimum integer such that $f_{\text{loss}}(t,p,\delta)\geq tN$ and
\begin{align}\label{f1}
        D\left(\frac{f_{\text{loss}}(t,p,\delta)+1}{t}\bigg\| N\right) & \geq \frac{1}{t}\log\left(\frac{2}{\delta}\right).
\end{align}
Finally, one can check that these have our desired structure,
\begin{align}
\begin{cases}
    \langle\phi_+^0|\phi_+^1\rangle  = \langle\phi_+^1|\phi^0_+\rangle = \langle\phi^0_-|\phi_-^1\rangle  = \langle\phi^1_-|\phi^0_-\rangle = 2\nu_{\text{loss}}^2-1, \\
    \langle\phi_+|\phi_+^0\rangle = \langle\phi_+|\phi_+^1\rangle = \langle\phi_-|\phi_-^0\rangle = \langle\phi_-|\phi^1_-\rangle = \nu_{\text{loss}},
\end{cases}
\end{align}
where
\begin{align}
        \nu_{\text{loss}} & \equiv \sqrt{\frac{2p^2-2p+1}{3p^2-2p+1}}.
\end{align}
Assume that we sample $m_1$ $|\phi_+\rangle$, $m_2$ $|\phi_-\rangle$ and $t-m_1-m_2$ Gaussian states, and $m_1$ and $m_2$ satisfy $m_1+m_2 > 16\nu_{\text{loss}}^{-2(m_1+m_2)}\delta^{-2}$. We remark that if $m_1+m_2$ does not satisfy this condition, then the next step is not applied. Without loss of generality, it is sufficient to discuss $|\phi_+^{\otimes m_1}\rangle|\phi_-^{\otimes m_2}\rangle$, because we are focusing on the fidelity. By applying Lem.~\ref{Low rank}, we can obtain a low-rank Gaussian superposition state that has high fidelity with the sampled state.
\begin{lemma}\label{Fermion particle loss}
    There is a state $|\mu\rangle$ such that 
    \begin{align}
            \left|\left(\langle\phi_+^{\otimes m_1}|\langle\phi_-^{\otimes m_2}|\right)|\mu\rangle\right|^2 & \geq 1-\frac{\delta^2}{4}
    \end{align}
    and $|\mu\rangle$ can be expressed as a superposition of $2^l$ Gaussian states satisfying $2^l\leq 16\nu_{\textup{loss}}^{-2(m_1+m_2)}\delta^{-2}$.
\end{lemma}
By the truncation method, the worst case is when we sample $f_{\text{loss}}(t,p,\delta)$ non-Gaussian states. Therefore, any sampled state from $\hat{\rho}_{\text{loss}}^{\otimes t}$ can be approximately replaced with a pure state with at most $O(\nu_{\text{loss}}^{-2f_{\text{loss}}(t,p,\delta)}\delta^{-2})$ Gaussian rank, and we can find it in at most $O(\nu_{\text{loss}}^{-2f_{\text{loss}}(t,p,\delta)}\delta^{-4})$ time. We again remark that the remaining work is just a few phase-sensitive Gaussian computations~\cite{Gaussian_rank2}. In short, the application of our strategy to the particle loss case can be summarized as follows:


\begin{theorem}
    Consider the noisy fermionic scheme with particle loss and let $\delta = \Omega(1/\textup{poly}(n))$. There is a classical algorithm for approximately simulating the circuit in at most $O(\nu_{\textup{loss}}^{-2f_{\textup{loss}}(t,p,\delta)}\textup{poly}(n))$ time, up to $\delta$ TVD.
\end{theorem}

In the end, we investigate when this noisy scheme with particle loss collapses into a classical computational regime. As in the qubit case, we show which condition on $p$ (or $t$) makes this noisy scheme collapse into a classical-simulable regime when $t$ (or $p$) is fixed. 

\begin{theorem}\label{Thm:Fermion_poly}
    Consider the noisy scheme with particle loss and let $\delta = \Omega(1/\textup{poly}(n))$. Assume that the number of injected magic states $t$ is given. If the transmission rate $p$ satisfies
    \begin{align}
    p & = O\left(\sqrt{\frac{\log n}{t}}\right),
    \end{align}
    then the time complexity becomes $O(\textup{poly}(n))$. Conversely, assume that the transmission rate $p$ is given. If the number of injected magic states $t$ satisfies
    \begin{align}
    t & = O\left(\frac{\log n}{p^2}\right),
    \end{align}
    then the time complexity becomes $O(\textup{poly}(n))$.
\end{theorem}

A detailed proof is provided in App.~\ref{App:Poly}. In this case, one can have the following intuitions: the noisy scheme may collapse into a classical simulable regime when $p$ is near zero, because every state becomes the vacuum state, which does not induce any computational power, when $p = 0$. Or, the same instance may happen when $t$ is near $O(\log n)$, because this scheme can be classically simulated regardless of $p$ when $t = O(\log n)$. Thm.~\ref{Thm:Fermion_poly} aligns with these intuitions. When $t$ is fixed, the threshold on $p$ decreases with $t$. As in the qubit case, it indicates that the more magic states are injected, the more tolerant the computational power becomes to physical noise. The extreme case is also consistent. When $t=0$, the classically simulable region extends over all $p$, which is a similar result to the qubit case.

The opposite case shows similar behavior. As desired, the threshold on $t$ decreases as $p$ approaches zero, corresponding to the extremely noisy case. In the opposite direction, i.e., as $p$ approaches $1$, the threshold on $t$ approaches $O(\log n)$, which aligns with our intuition.

This analysis can be extended to other physical noises. More concretely, we also investigate the effect of dephasing noise. We remark that our described ensemble in this case consists of solely non-Gaussian states. Therefore, we always sample $t$ non-Gaussian states from this description and cannot apply the truncation method. However, we show Gaussian decompositions of these states to apply Lem.~\ref{Low rank}. Denote $|\omega_\pm\rangle \equiv C_\pm|0011\rangle + C_\mp|1100\rangle$. Although our proposed description consists of two non-Gaussian states, each state can also be decomposed as two Gaussian states,
\begin{align}
        |\omega_\pm\rangle & \equiv \sqrt{E}\left(|\omega^0_\pm\rangle+|\omega^1_\pm\rangle\right),
\end{align}
where, for $j=0,1$,
\begin{align}\label{Eq:Gaussian in Fermion_dephasing}
    |\omega^j_{\pm}\rangle  & \equiv \frac{1}{\sqrt{E}}(C_{\pm}|0011\rangle + C_{\mp}|1100\rangle \nonumber \\
    & + (-1)^j D |0000\rangle + (-1)^jD|1111\rangle),
\end{align}
$D \equiv (1-2p)^2/\sqrt{2}$ and $E \equiv (1-2p)^4+1$. We check that these states are Gaussian in App.~\ref{Check Gaussian}. Furthermore, one can check that these states have our desired structures,
\begin{align}
    \begin{cases}\langle\omega_+^0|\omega_+^1\rangle  = \langle\omega_+^1|\omega^0_+\rangle  = \langle\omega^0_-|\omega_-^1\rangle  = \langle\omega^1_-|\omega^0_-\rangle = 2\nu_{\text{dep}}^2-1, \\
    \langle\omega_+|\omega_+^0\rangle = \langle\omega_+|\omega_+^1\rangle = \langle\omega_-|\omega_-^0\rangle = \langle\omega_-|\omega^1_-\rangle = \nu_{\text{dep}},
    \end{cases}
\end{align}
where
    \begin{align}\label{Eq:nu2}
            \nu_{\text{dep}} & \equiv \sqrt{\frac{1}{1+(1-2p)^4}}.
    \end{align}
Now assume that we sample $m_1$ $|\omega_+\rangle$ and $m_2$ $|\omega_-\rangle$, where $m_1+m_2 = t$. As we said before, we assume $t > 16\nu_{\text{dep}}^{-2t}\delta^{-2}$. By applying Lem.~\ref{Low rank}, we can obtain a low-rank superposition of Gaussian states that has high fidelity with $|\omega_+^{\otimes m_1}\rangle|\omega_-^{\otimes m_2}\rangle$.
\begin{lemma}\label{Fermion Dephasing}
    There is a state $|\mu\rangle$ such that
    \begin{align}
            \left|\left(\langle\omega_+^{\otimes m_1}|\langle\omega_-^{\otimes m_2}|\right)|\mu\rangle\right|^2 & \geq 1-\frac{\delta^2}{4}
    \end{align}
    and $|\mu\rangle$ can be expressed as a superposition of $2^l$ Gaussian states satisfying $2^l\leq 16\nu_{\textup{dep}}^{-2t}\delta^{-2}$.
\end{lemma}

In short, the application of our strategy to this case is characterized as follows:
\begin{figure*}[t!]
    \centering
    \includegraphics[width=1.0\linewidth]{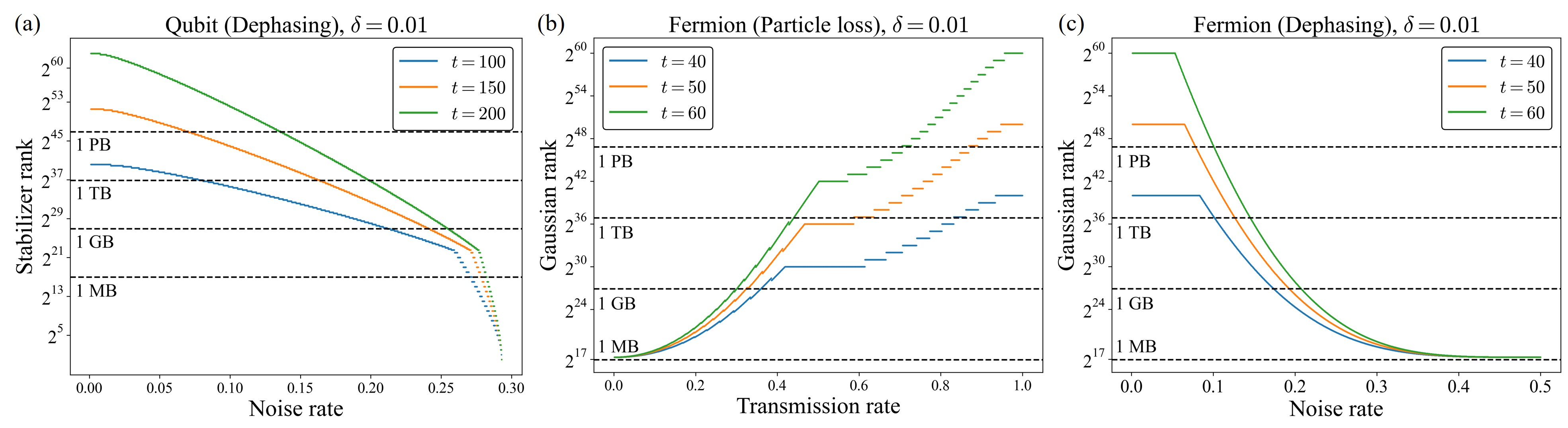}
    \caption{The resourceless rank of the output state (log scale) over the noise rate $p$. We set $\delta=0.01$ and investigate the various numbers of qubits (or modes). Each dashed line indicates the number of state coefficients that can be stored for a given amount of data storage. (a) Qubit case with the dephasing noise. (b) Fermion case with the particle loss. (c) Fermion case with the dephasing noise.}
    \label{fig:Numerical}
\end{figure*}

\begin{theorem}
    Consider the noisy fermionic scheme with dephasing noise and let $\delta = \Omega(1/\textup{poly}(n))$. There is a classical algorithm for approximately simulating the circuit in at most $O(\nu_{\textup{dep}}^{-2t}\textup{poly}(n))$ time, up to $\delta$ TVD.
\end{theorem}

Moreover, we investigate when this noisy scheme with dephasing noise collapses into a classical computational regime, for both cases where $p$ or $t$ is given. As we discussed previously, it is sufficient to discuss only $p\in[0,1/2]$, because the dephasing noise has a symmetry with the axis $p = 1/2$. Therefore, we restrict the noise rate to the interval $[0,1/2]$ in the remainder of the discussion.
\begin{theorem}\label{Thm:Fermion_poly2}
    Consider the noisy scheme with dephasing noise and let $\delta = \Omega\left(1/\textup{poly}(n)\right)$. Assume that the number of injected magic states $t$ is given. If the noise rate $p$ satisfies
    \begin{align}
    \frac{1}{2} - p & = O\left(\left(\frac{\log n}{t}\right)^{\frac{1}{4}}\right),
    \end{align}
    then the time complexity becomes $O(\textup{poly}(n))$. Conversely, assume that the noise rate $p$ is given. If the number of injected magic states $t$ satisfies
    \begin{align}
    t & = O\left(\frac{\log n}{\left(\frac{1}{2}- p\right)^4}\right),
    \end{align}
    then the time complexity becomes $O(\textup{poly}(n))$.
\end{theorem}

A detailed proof is provided in App.~\ref{App:Poly}. We mention that the critical point is when $p = 1/2$. As desired, Thm.~\ref{Thm:Fermion_poly2} also shows similar results to the prior discussions. For the result related to the threshold on $p$, the smaller $t$ is, the more the classical simulable regime extends to the interval of $p$. Conversely, for the threshold on $t$, the closer $p$ approaches the critical point, the larger the threshold becomes. Finally, these results also recover the noiseless limits, as discussed previously.

In short, we analytically characterize how particle loss and dephasing noise degrade the computational power of a fermionic circuit that allows universal quantum computation. Moreover, we identify noise and resource thresholds in terms of $p$ and $t$ at which the noisy scheme collapses into a classically simulable regime. Finally, we verify that our analysis is consistent with the noiseless limit described by the Gaussian computation.

\section{Numerical results}\label{Sec:Computational boundary}
In the preceding sections, we introduced a classical simulation algorithm for a quantum circuit with noisy magic states and demonstrated its application to several representative schemes. As discussed earlier, the output of our algorithm is a low-rank superposition of resourceless states, which dominates the computational cost of the remaining algorithm, a few phase-sensitive resourceless computations. In this section, we quantify the resourceless rank of the output state in the worst case under several situations, varying both the number of injected magic states and the noise rate. More precisely, we plot the resourceless rank of our output state in the worst case over the noise rate $p$, with the various number of qubits (or modes) $n$ and a fixed $\delta=0.01$~(see Fig.~\ref{fig:Numerical}).

We remark that our strategy has a decision point: whether to use the second subroutine described in Lem.~\ref{Low rank}. Therefore, one can see that there is a transition of the shape of functions in Fig.~\ref{fig:Numerical}. For instance, in Fig.~\ref{fig:Numerical}~(b), the function becomes discrete after the transition, because the Gaussian rank of the worst case without the second subroutine is $2^{f_{\text{loss}}(t,p,\delta)}$ and $f_{\text{loss}}(t,p,\delta)$ is discrete over $p$.

As desired, in the qubit case, indicated by Fig.~\ref{fig:Numerical}~(a), the stabilizer rank approaches $1$ as the noise rate goes to the threshold. On the other hand, the fermion cases, indicated by Fig.~\ref{fig:Numerical}~(b) and~(c), approach a larger value, $2^{17}$. The reason is that the probabilities of sampling a non-Gaussian are always over $0$. More precisely, the probability in the former case is at least $1/2$, and the probability in the latter case is always $1$. It implies that the second subroutine dominates near the extreme point, so that the Gaussian rank of the worst case is described as $16\nu_{\text{loss}}^{-2f_{\text{loss}}(t,p,\delta)}\delta^{-2}$ (or $16\nu_{\text{dep}}^{-2t}\delta^{-2}$). Therefore, the lower bound $16\delta^{-2}$ occurs~(In this case, $16\delta^{-2}\simeq 2^{17}$), because $\nu_{\text{loss}}$ and $\nu_{\text{dep}}$ is limit to $1$ as $p$ approaches to the extreme point.

Furthermore, we represent the number of state coefficients, which are complex numbers, that can be stored for a given amount of data storage~(see dashed lines in Fig.~\ref{fig:Numerical}). More precisely, we can derive the possible number of stored complex numbers by dividing a given amount of data storage by the required memory for each complex number. We note that 8 bytes are required for storing a single complex number with the complex64 data type. 

For instance, in Fig.~\ref{fig:Numerical}~(a), if 1 TB data storage is given, then we can store coefficients of the results of our algorithm when $p$ is greater than approximately $0.0809$ and $t = 100$. In contrast, when $t = 200$, the possible region becomes tight; we can store the coefficients of the results when $p$ is greater than approximately $0.1993$, given the memory.

Similarly, in Fig.~\ref{fig:Numerical}~(b), 1 TB data storage can save the cases when $p$ is less than approximately $0.8320$ if $t = 40$. Whereas, if $t = 60$, the possible region of $p$ is tighter; the given memory can save the cases when $p$ is less than approximately $0.4402$. Fig.~\ref{fig:Numerical}~(c) also shows a similar behavior; 1 TB data storage can save the coefficients of the results when $p$ is over than approximately $0.1015$ if $t = 40$, however, if $t = 60$, it can save the coefficients of the cases when $p$ is over than approximately $0.1453$.
 
We emphasize that more storage is required for the actual simulation, as we must store classical descriptions of states. However, these values are polynomial (over $n$) multiples of the amount of memory required to store state coefficients, since states are resourceless.

\section{Discussion}\label{Sec:Discussion}
We have studied the classical simulability of a quantum circuit with injected noisy magic states, which permits universal quantum computation in a noiseless environment. To achieve our goal, we propose a novel classical simulation algorithm for this noisy circuit and quantify its computational power by characterizing the simulation cost of our algorithm. More precisely, our primary strategy is to transform a simulation with density operators into one with a low-rank resourceless superposition state. Hence, the output of our algorithm is a classical description of a low-rank resourceless superposition state, and our simulation costs can be characterized as its resourceless rank. Finally, since the simulation costs can be expressed in a closed form, we can identify classical simulation regimes that require polynomial-scale costs to simulate a quantum circuit when we apply our algorithm to a specific situation.

Our strategy comprises three parts: finding a suitable ensemble, sampling and truncation methods, and generating a low-rank resourceless superposition state. We begin by focusing on the probabilistic structure of density operators and show the equivalence between sampling input states and the original scheme. Furthermore, since it can be interpreted as a binomial sampling procedure in the computational cost aspect, we truncate the cases that require large computational costs, but have low probabilities. Finally, although its sampled input state can generally remain a high-rank resourceless superposition state, we show that we can reduce its rank with high fidelity by benchmarking Ref.~\cite{Bravyi}. Throughout our algorithm, constructing a well-designed ensemble description is also a crucial part. Hence, we illustrate the conditions of mixed states tailored to our algorithm.

Furthermore, we show how our approach can be applied to specific situations. We focus on two universal quantum computing systems, the qubit and fermionic systems. We model a noisy channel as dephasing noise in the qubit system and as particle loss and dephasing noise in the fermionic system. Furthermore, we investigate when its noisy situation collapses into a classical simulable region over two parameters, the number of injected magic states $t$ and the noise rate (or transmission rate) $p$. Finally, we numerically measure the resourceless ranks of the worst case when our algorithm is applied to various situations.

There are several open questions still. One of the natural open questions is the optimal decomposition of magic states. For example, stabilizer and Gaussian extent are the quantities for evaluating stabilizer and Gaussian state decomposition~\cite{Gaussian_rank1, Gaussian_rank2, Stabilizer_rank, Bosonic_Gaussian_rank1, Bosonic_Gaussian_rank2}. We anticipate that if one can find its extent-optimal decomposition, then this decomposition will yield better performance. Another open question is the optimal choice of an ensemble of noisy magic states. As discussed previously, the performance of our algorithm depends on the choice of the ensemble. We anticipate that one can find an ensemble more tailored to our algorithm for specific parameters. Additionally, one can combine our strategy and another approach to simulate these noisy schemes classically, such as the matrix product state and the matrix product operator~\cite{Classical_BS1, Classical_GBS1, Classical_BS3, Classical_RCS3, MPS, MPS1, MPS2, MPS3}. Lastly, we anticipate that our strategy can be applied to other various scenarios, such as other magic states or systems.

\acknowledgments
We are grateful to Su-un Lee for helpful discussions on magic state distillation. J.H., S.P., and C.O. were supported by the National
Research Foundation of Korea Grants (No. RS-2024-00431768 and No. RS-2025-00515456) funded by the Korean government (Ministry of Science and ICT~(MSIT)) and the Institute of Information \& Communications Technology Planning \& Evaluation (IITP) Grants funded by the Korea government (MSIT) (No. IITP-2025-RS-2025-02283189 and IITP-2025-RS-2025-02263264).
This work was supported by Global Partnership Program of Leading Universities in Quantum Science and Technology (RS-2025-08542968) through the National Research Foundation of Korea(NRF) funded by the Korean government (Ministry of Science and ICT(MSIT)).

\bibliography{References}

\appendix

\section{Proof of Lem.~\ref{Low rank}}\label{App:Proof of Low rank Thm}
In this section, we give a detailed proof of Lem.~\ref{Low rank}. We again emphasized that Lem.~\ref{Low rank} is inspired by Ref.~\cite{Bravyi}.
\begin{proof}
    By the definition, $|\Psi\rangle$ can be described as 
    \begin{align}
            |\Psi\rangle & = \frac{1}{(2\nu)^m}\sum_{\mathbf{x}\in\mathbb{Z}_2^m} |\psi_1^{x_1}\rangle\cdots|\psi_m^{x_m}\rangle,
    \end{align}
    where $\mathbf{x} = (x_1,\dots,x_m)$. Let $\mathcal{L}$ be a subspace of $\mathbb{Z}_2^m$ of dimension $l$, and define its corresponding state $|\mathcal{L}\rangle$ as
    \begin{align}
            |\mathcal{L}\rangle & = \frac{1}{\sqrt{2^l Z(\mathcal{L})}}\sum_{\mathbf{x}\in\mathcal{L}} |\psi_1^{x_1}\rangle\cdots|\psi_m^{x_m}\rangle,
    \end{align}
    where $Z(\mathcal{L})$ is the normalization factor. Since $\langle\psi_i^0|\psi_i^1\rangle=\langle\psi_i^1|\psi_i^0\rangle = 2\nu^2-1$ for all $i=1,\dots, m$, we can compute $Z(\mathcal{L})$ directly
    \begin{align}
            Z(\mathcal{L}) & = \frac{1}{2^l}\sum_{\mathbf{x}\in\mathcal{L}}\sum_{\mathbf{y}\in\mathcal{L}} \langle \psi_1^{x_1}|\psi_1^{y_1}\rangle\cdots\langle\psi_m^{x_m}|\psi_m^{y_m}\rangle \\
            & = \frac{1}{2^l}\sum_{\mathbf{x}\in\mathcal{L}}\sum_{\mathbf{y}\in\mathcal{L}}\left(2\nu^2-1\right)^{H(\mathbf{x}\oplus \mathbf{y})} \\
            & = \sum_{\mathbf{x}\in\mathcal{L}}\left(2\nu^2-1\right)^{H(\mathbf{x})},
    \end{align}
    where $H(\mathbf{x})$ is the Hamming weight of $\mathbf{x}$ and $\oplus $ is the bitwise addition. One can easily check that $\langle\psi_i|\psi_i^0\rangle = \langle\psi_i|\psi_i^1\rangle = \nu$. Then, we obtain
    \begin{align}
            |\langle\Psi|\mathcal{L}\rangle|^2 & = \frac{1}{2^lZ(\mathcal{L})}\left|\sum_{\mathbf{x}\in\mathcal{L}}\langle\psi_1|\psi_1^{x_1}\rangle\cdots \langle\psi_m|\psi_m^{x_m}\rangle\right|^2 \\
            & = \frac{2^l \nu^{2m}}{Z(\mathcal{L})}.
    \end{align}
    Now, let $\mathcal{L}$ be a uniformly chosen subspace of $\mathbb{Z}_2^m$ of dimension $l$. By linearity, we can compute the expectation value of $Z(\mathcal{L})$,
    \begin{align}
            \mathbb{E}[Z(\mathcal{L})] & = 1 + \sum_{\mathbf{x}\in\mathbb{Z}_2^m\backslash\{0\}}\left(2\nu^2-1\right)^{H(\mathbf{x})}\mathbb{E}\left[\chi_{\mathcal{L}}(\mathbf{x})\right] \\
            & = 1 + \frac{2^l-1}{2^m-1}\sum_{\mathbf{x}\in\mathbb{Z}_2^m\backslash\{0\}}\left(2\nu^2-1\right)^{H(\mathbf{x})} \\
            & = 1 + \frac{2^l-1}{2^m-1}\left\{(2\nu^2)^m - 1\right\} \\
            & \leq 1 + 2^l\nu^{2m},
    \end{align}
    where $\chi_{\mathcal{L}}$ is the indicator function of $\mathcal{L}$ such that $\chi_{\mathcal{L}}(\mathbf{x})=1$ if $\mathbf{x}\in\mathcal{L}$, and $\chi_{\mathcal{L}}(\mathbf{x})=0$ otherwise. It implies that there exists at least one $\mathcal{L}$ such that $Z(\mathcal{L})\leq 1+2^l\nu^{2m}$. Fix $l$ such that $4\geq2^l\nu^{2m}\delta_2\geq2$. By Markov's inequality, we have
    \begin{align}
            & \textrm{Pr}\left[Z(\mathcal{L})\geq \left(1+2^l\nu^{2m}\right)\left(1+\frac{\delta_2}{2}\right)\right] \\
            & \leq \frac{\mathbb{E}\left[Z(\mathcal{L})\right]}{\left(1+2^l\nu^{2m}\right)\left(1+\frac{\delta_2}{2}\right)} \\
            & \leq 1- \frac{\delta_2}{2+\delta_2}.
    \end{align}
    Hence, we can obtain $\mathcal{L^*}$ such that $Z(\mathcal{L^*}) \leq \left(1+2^l\nu^{2m}\right)\left(1+\frac{\delta_2}{2}\right)$ with a constant probability by randomly choosing $O(\delta^{-1}_2)$ subspaces. It follows that
    \begin{align}
            \left|\langle\Psi|\mathcal{L}^*\rangle\right|^2 & = \frac{2^l\nu^{2m}}{Z(\mathcal{L}^*)} \\
            & \geq \frac{2^l\nu^{2m}}{\left(1+2^l\nu^{2m}\right)\left(1+\frac{\delta_2}{2}\right)} \\
            & = \frac{1}{\left(1+2^{-l}\nu^{-2m}\right)\left(1+\frac{\delta_2}{2}\right)} \\
            & \geq \frac{1}{\left(1+\frac{\delta_2}{2}\right)^2} \\
            & \geq 1-\delta_2.
    \end{align}
    Recall that $4\geq2^l\nu^{2m}\delta_2\geq 2$, and it implies that $2^l\leq 4\nu^{-2m}\delta_2^{-1}$.
\end{proof}

\section{Derivation of ensembles in the qubit system}\label{App:Optimization}
In this section, we show how to derive the ensembles described in Sec.~\ref{Sec:Qubit}. First of all, we start at the observation of the extreme point $p=1/2$. Remark that our noisy state $\hat{\rho}_{\text{qubit}}$ is given by
\begin{align}
        \hat{\rho}_{\text{qubit}} & = \begin{pmatrix}
            \cos^2(\frac{\pi}{8}) & (1-2p)\cos(\frac{\pi}{8})\sin(\frac{\pi}{8}) \\
            (1-2p)\cos(\frac{\pi}{8})\sin(\frac{\pi}{8}) & \sin^2(\frac{\pi}{8})
        \end{pmatrix}.
\end{align}
By putting $p=1/2$ into $\hat{\rho}_{\text{qubit}}$, the off-diagonal terms vanish
\begin{align}
        \hat{\rho}_{\text{qubit}}\big|_{p=\frac{1}{2}} & = \begin{pmatrix}
            \cos^2(\frac{\pi}{8}) & 0 \\ 
            0 & \sin^2(\frac{\pi}{8}) 
        \end{pmatrix} \\
        & = \cos^2\left(\frac{\pi}{8}\right)|0\rangle\langle0| + \sin^2\left(\frac{\pi}{8}\right)|1\rangle\langle 1|,
\end{align}
so that it can be expressed as an ensemble consisting of only stabilizer states. Since $\hat{\rho}_{\text{qubit}}$ is continuous over $p$, it should have only small off-diagonal terms near the extreme point. From this observation, we add a small $|+\rangle$ and $|-\rangle$ to supply small off-diagonal terms. More precisely, we set the following problem,
\begin{align}
        \hat{\rho}_{\text{qubit}} & = \begin{pmatrix}
            \cos^2(\frac{\pi}{8}) & (1-2p)\cos(\frac{\pi}{8})\sin(\frac{\pi}{8}) \\
            (1-2p)\cos(\frac{\pi}{8})\sin(\frac{\pi}{8}) & \sin^2(\frac{\pi}{8})
        \end{pmatrix} \\
        & = x_1|0\rangle\langle0| + x_2|1\rangle\langle1| + x_3|+\rangle\langle+| + x_4|-\rangle\langle-|,
\end{align}
such that
\begin{enumerate}
    \item $x_1+x_2+x_3+x_4 = 1$
    \item $x_i\geq 0$ for all $i=1,\dots,4$.
\end{enumerate}
This problem can be transformed into
\begin{enumerate}
    \item $x_1+\frac{1}{2}x_3 + \frac{1}{2}x_4 = \cos^2(\frac{\pi}{8})$
    \item $x_2 + \frac{1}{2}x_3 + \frac{1}{2}x_4 = \sin^2(\frac{\pi}{8})$
    \item $\frac{1}{2}(x_3-x_4) = (1-2p)\cos(\frac{\pi}{8})\sin(\frac{\pi}{8})$
    \item $x_i\geq0$ for all $i=1,\dots, 4$.
\end{enumerate}
Since we have three equalities, we can set a single free parameter. By setting a free parameter as $x_1$, we have
\begin{enumerate}
    \item $x_2 = x_1 + \sin^2(\frac{\pi}{8})-\cos^2(\frac{\pi}{8})$
    \item $x_3 = -x_1+\cos^2(\frac{\pi}{8})+(1-2p)\cos(\frac{\pi}{8})\sin(\frac{\pi}{8})$
    \item $x_4 = -x_1+\cos^2(\frac{\pi}{8})-(1-2p)\cos(\frac{\pi}{8})\sin(\frac{\pi}{8})$.
\end{enumerate}
Since $x_i\geq0$ for all $i=1,\dots,4$, we obtain
\begin{enumerate}
    \item $x_1\geq \cos^2(\frac{\pi}{8})-\sin^2(\frac{\pi}{8})$
    \item $x_1\leq \cos^2(\frac{\pi}{8})+(1-2p)\cos(\frac{\pi}{8})\sin(\frac{\pi}{8})$
    \item $x_1\leq \cos^2(\frac{\pi}{8})-(1-2p)\cos(\frac{\pi}{8})\sin(\frac{\pi}{8})$
    \item $x_1\geq 0$.
\end{enumerate}
One can easily check that $\cos^2(\pi/8)-\sin^2(\pi/8)>0$. If $p\leq1/2$, then the following inequality should hold,
\begin{align}
        & \cos^2\left(\frac{\pi}{8}\right)-\sin^2\left(\frac{\pi}{8}\right) \nonumber \\
        & \leq \cos^2\left(\frac{\pi}{8}\right)-(1-2p)\cos\left(\frac{\pi}{8}\right)\sin\left(\frac{\pi}{8}\right),
\end{align}
so that $p\geq \left(1-\tan(\pi/8)\right)/2$. Similarly, if $p\geq1/2$, then the following inequality should hold,
\begin{align}
        & \cos^2\left(\frac{\pi}{8}\right)-\sin^2\left(\frac{\pi}{8}\right) \nonumber \\
        & \leq \cos^2\left(\frac{\pi}{8}\right)+(1-2p)\cos\left(\frac{\pi}{8}\right)\sin\left(\frac{\pi}{8}\right),
\end{align}
so that $p\leq (1+\tan(\pi/8))/2$. Hence, we can conclude that $\hat{\rho}_{\text{qubit}}$ can be expressed as ensembles consisting of only stabilizer states on $p\in[(1-\tan(\pi/8))/2,(1+\tan(\pi/8))/2]$. Finally, in order to get entire ensembles, set $x_4=0$ when $p\in [(1-\tan(\pi/8))/2,1/2]$. Then, the ensemble is described as 

\begin{center}
\begin{tabular}{@{\hspace{1.45cm}}c@{\hspace{1.45cm}\vrule\hspace{1.45cm}}c@{\hspace{1.45cm}}}
  \textbf{State} & \textbf{Probability} \\ \hline
  $|0\rangle$ & $\frac{1}{2}+\frac{1}{\sqrt{2}}p$           \\
  $|1\rangle$ & $\frac{1-\sqrt{2}}{2} + \frac{1}{\sqrt{2}}p$                         \\
  $| +\rangle$& $\frac{1}{\sqrt{2}}-\sqrt{2}p$
\end{tabular}
\end{center}

Similarly, as setting $x_3=0$ when $p\in [1/2, (1+\tan(\pi/8))/2]$, the ensemble is described as

\begin{center}
\begin{tabular}{@{\hspace{1.45cm}}c@{\hspace{1.45cm}\vrule\hspace{1.45cm}}c@{\hspace{1.45cm}}}
  \textbf{State} & \textbf{Probability} \\ \hline
  $|0\rangle$ & $\frac{1+\sqrt{2}}{2}-\frac{1}{\sqrt{2}}p\ $           \\
  $|1\rangle$ & $\frac{1}{2} - \frac{1}{\sqrt{2}}p\ $                         \\
  $| -\rangle$& $\sqrt{2}p-\frac{1}{\sqrt{2}}$.
\end{tabular}
\end{center}

Now, consider the region outside $[(1-\tan(\pi/8))/2, (1+\tan(\pi/8))/2]$. We start at the noiseless case $p=0$. By the definition, $\hat{\rho}_{\text{qubit}}|_{p=0}$ can be expressed as $|H\rangle$ entirely. Since $\hat{\rho}_{\text{qubit}}$ is continuous over $p$, we probably can find an ensemble connecting the points $p=0$ and $p=(1-\tan(\pi/8))/2$. More precisely, we set the following problem
\begin{align}
        \hat{\rho}_{\text{qubit}} & = \begin{pmatrix}
            \cos^2(\frac{\pi}{8}) & (1-2p)\cos(\frac{\pi}{8})\sin(\frac{\pi}{8}) \\
            (1-2p)\cos(\frac{\pi}{8})\sin(\frac{\pi}{8}) & \sin^2(\frac{\pi}{8})
        \end{pmatrix} \\
        & = x_1|0\rangle\langle 0| + x_2|1\rangle\langle 1| + x_3|+\rangle\langle+| +x_4|H\rangle\langle H|,
\end{align}
such that
\begin{enumerate}
    \item $x_1+x_2+x_3+x_4 = 1$
    \item $x_i\geq 0$ for all $i=1,\dots,4$.
\end{enumerate}
This problem can be transformed into
\begin{enumerate}
    \item $x_1+\frac{1}{2}x_3 + \frac{3+2\sqrt{2}}{2(2+\sqrt{2})}x_4 = \cos^2(\frac{\pi}{8})$ 
    \item $x_2+\frac{1}{2}x_3 + \frac{1}{2(2+\sqrt{2})}x_4 = \sin^2(\frac{\pi}{8})$
    \item $\frac{1}{2}x_3 + \frac{\sqrt{2}+1}{2(2+\sqrt{2})}x_4 = (1-2p)\cos(\frac{\pi}{8})\sin(\frac{\pi}{8})$
    \item $x_i\geq0$ for all $i=1,\dots, 4$.
\end{enumerate}
Since we have three equalities, we can set a single free parameter. By setting a free parameter as $x_4$, we have
\begin{enumerate}
    \item $x_1 = -\frac{1}{2}x_4 + \cos^2(\frac{\pi}{8})-(1-2p)\cos(\frac{\pi}{8})\sin(\frac{\pi}{8})$
    \item $x_2 = \frac{1}{\sqrt{2}(2+\sqrt{2})}x_4+\sin^2(\frac{\pi}{8})-(1-2p)\cos(\frac{\pi}{8})\sin(\frac{\pi}{8})$
    \item $x_3 = -\frac{\sqrt{2}+1}{2+\sqrt{2}}x_4 + 2(1-2p)\cos(\frac{\pi}{8})\sin(\frac{\pi}{8})$.
\end{enumerate}
Since $x_i\geq0$ for all $i=1,\dots,4$, these relations imply
\begin{enumerate}
    \item $x_4\leq 2\left[\cos^2(\frac{\pi}{8})-(1-2p)\cos(\frac{\pi}{8})\sin(\frac{\pi}{8})\right]$
    \item $x_4\geq \sqrt{2}(2+\sqrt{2})\left[(1-2p)\cos(\frac{\pi}{8})\sin(\frac{\pi}{8})-\sin^2(\frac{\pi}{8})\right]$
    \item $x_4\leq\frac{2(2+\sqrt{2})}{1+\sqrt{2}}(1-2p)\cos(\frac{\pi}{8})\sin(\frac{\pi}{8})$
    \item $x_4\geq0$
\end{enumerate}
One can check that
\begin{align}
        0 \leq \sqrt{2}(2+\sqrt{2})\left[(1-2p)\cos\left(\frac{\pi}{8}\right)\sin\left(\frac{\pi}{8}\right)-\sin^2\left(\frac{\pi}{8}\right)\right],
\end{align}
on $p\in[0, (1-\tan(\pi/8))/2]$. Furthermore, one can easily check that the conditions of $x_4$ are well-defined. To minimize the probability of sampling a non-stabilizer state, we set $x_4$ has the minimum value. In other words, by setting
\begin{align}
        x_4 & = \sqrt{2}(2+\sqrt{2})\left[(1-2p)\cos\left(\frac{\pi}{8}\right)\sin\left(\frac{\pi}{8}\right)-\sin^2\left(\frac{\pi}{8}\right)\right],
\end{align}
we obtain our desired ensemble on $p\in[0, (1-\tan(\pi/8))/2]$,

\begin{center}
\begin{tabular}{@{\hspace{1.4cm}}c@{\hspace{1.4cm}\vrule\hspace{1.4cm}}c@{\hspace{1.4cm}}}
  \textbf{State} & \textbf{Probability} \\ \hline
  $|0\rangle$ & $\left(1+\sqrt{2}\right)p$           \\
  $|+\rangle$ & $p$                         \\
  $| H\rangle$& $1-\left(2+\sqrt{2}\right)p$
\end{tabular}
\end{center}

In a similar way, we can also derive the desired ensemble on $p\in [(1+\tan(\pi/8))/2,1]$,
\begin{center}
\begin{tabular}{@{\hspace{1.4cm}}c@{\hspace{1.4cm}\vrule\hspace{0.9cm}}c@{\hspace{0.9cm}}}
  \textbf{State} & \textbf{Probability} \\ \hline
  $|0\rangle$ & $\left(1+\sqrt{2}\right)\left(1-p\right)\ $           \\
  $|-\rangle$ & $1-p\ $                         \\
  $|H^\prime\rangle$& $1-\left(2+\sqrt{2}\right)\left(1-p\right)$,
\end{tabular}
\end{center}

where $|H^\prime\rangle = \cos(\pi/8)|0\rangle - \sin(\pi/8)|1\rangle$. We again emphasize that the whole ensembles have a symmetry with the axis $p=1/2$.

\section{Analysis of time complexity of qubit case}\label{App:Qubit_Time complexity}
In Sec.~\ref{Sec:Qubit}, we illustrate our focused quantum circuit that allows universal quantum computation in the qubit system, and investigate how physical noise destroys its computational power. As we discussed previously, our algorithm is applied to the ensemble defined on $p\in[0, (1-\tan(\pi/8))/2]$, because the ensemble on $p\in[(1-\tan(\pi/8))/2, 1/2]$ can be represented as a convex summation of solely stabilizer states. Remark that our situation has a symmetry with the axis $p=1/2$, so that we restrict our discussion to $p\in[0,1/2]$. In the former region, the time complexity of our algorithm goes to a classically simulable regime as $p$ goes to the boundary. Hence, this noisy scheme may collapse into a classical simulable regime when $p$ is close to $(1-\tan(\pi/8))/2$, even if the noisy magic state cannot be expressed as a convex summation of stabilizer states. 

In this section, we clarify a boundary that this scheme collapses into a classical simulable regime. Remark that the time complexity of our algorithm on $p\in[0, (1-\tan(\pi/8))/2]$ is $O(\nu_{\text{qubit}}^{-2f_{\text{qubit}}(t,p,\delta)}\text{poly}(n))$,
where $\nu_\text{qubit} = \cos(\pi/8)$. Therefore, to investigate when this time complexity becomes polynomial scale, it is sufficient to clarify when $f_{\text{qubit}}(t,p,\delta)$ becomes $O(\log n)$.

\begin{theorem}\label{Thm:Qubit_App_poly1}
    Consider the ensemble defined on $p \in [0, (1-\tan(\pi/8))/2]$, described in Sec.~\ref{Sec:Qubit}. Assume that the number of injected magic states $t$ is given. For $\delta = \Omega(1/\textup{poly}(n))$, if the noise rate $p$ satisfies
    \begin{align}
        p & = \frac{1-\tan\left(\frac{\pi}{8}\right)}{2}\left(1-O\left(\frac{\log n}{t}\right)\right),
    \end{align}
    then the time complexity becomes $O(\textup{poly}(n))$.
\end{theorem}
\begin{proof}
It is sufficient to show that $f_{\text{qubit}}(t,p,\delta)= O(\log n)$. Since $f_{\text{qubit}}(t,p,\delta)$ is the smallest number satisfying the conditions in Eq.~\eqref{f0}, we have
\begin{align}\label{Eq:qubit_poly1}
    D\left(\frac{f_{\text{qubit}}(t,p,\delta)}{t}\bigg\|1-(2+\sqrt{2})p\right) & \leq \frac{1}{t}\log\left(\frac{2}{\delta}\right).
\end{align}
For convenience, denote $a = f_{\text{qubit}}(t,p,\delta)/t$ and $r = 1- (2+\sqrt{2})p$. Then, the given condition of $p$ implies
\begin{align}\label{Eq:qubit_poly2}
    r &  = O\left(\frac{\log n}{t}\right).
\end{align}
Since $a\geq r$, we obtain
\begin{align}\label{Eq:qubit_poly3}
    D(a\| r) & \geq \frac{(a-r)^2}{2a}.
\end{align}
Combining Eq.~\eqref{Eq:qubit_poly1} and~\eqref{Eq:qubit_poly3}, we have
\begin{align}\label{Eq:qubit_poly4}
    |f_{\text{qubit}}(t,p,\delta)-tr| &\leq \sqrt{2f_{\text{qubit}}(t,p,\delta)\log(2/\delta)}.
\end{align}
It implies that
\begin{align}\label{Eq:Qubit_poly1_inequality1}
    & f_{\text{qubit}}(t,p,\delta) \nonumber \\
    &\leq tr + \log\left(\frac{2}{\delta}\right) + \sqrt{\left(tr+\log\left(\frac{2}{\delta}\right)\right)^2 - (tr)^2}.
\end{align}
It follows that $f_{\text{qubit}}(t,p,\delta) = O(\log n)$ if $\delta = \Omega(1/\text{poly}(n))$.
\end{proof}

\begin{corollary}
    Consider the ensemble defined on $p \in [0, (1-\tan(\pi/8))/2]$, described in Sec.~\ref{Sec:Qubit}. Assume that the noise rate $p$ is given. For $\delta = \Omega(1/\textup{poly}(n))$, if the number of injected magic states $t$ satisfies
    \begin{align}
        t & = O\left(\frac{\log n}{r}\right),
    \end{align}
    where $r = 1-(2+\sqrt{2})p$, then the time complexity becomes $O(\textup{poly}(n))$.
\end{corollary}
\begin{proof}
    The given condition implies $tr = O(\log n)$. Then, Eq.~\eqref{Eq:Qubit_poly1_inequality1} implies $f_{\text{qubit}}(t,p,\delta) = O(\log n)$ if $\delta = \Omega(1/\text{poly}(n))$.
\end{proof}

\section{Gaussianity of two mode even state}\label{Two mode}
Here, we put a simple proof of why every two-mode even state is Gaussian. By following~\cite{Gaussian}, we start at the definition of a Gaussian state. For future convenience, we define a $n$-mode state $\hat{\rho}_{\vec{\lambda}}$ as
\begin{align}
        \hat{\rho}_{\vec{\lambda}} & = \frac{1}{2^n}\prod_{j=1}^n\left(\hat{I}+\lambda_j\hat{\sigma}_j^z\right) = \frac{1}{2^n}\prod_{j=1}^n\left(\hat{I}-i\lambda_j\hat{c}_{2j-1}\hat{c}_{2j}\right),
\end{align}
where $\hat{\sigma}_j^z$ is the Pauli-$Z$ on $j$-th mode, $\hat{c}_{2j-1}$ and $\hat{c}_{2j}$ are Majorana operators defined as
\begin{align}
        \hat{c}_{2j-1}  = \hat{a}_j + \hat{a}_j^\dagger,~~~
        \hat{c}_{2j}  = -i\left(\hat{a}_j - \hat{a}_j^\dagger\right).
\end{align}
A Gaussian state is defined as a unitary transformation, consisting of a quadrature representation of Majorana operators, of $\hat{\rho}_{\vec{\lambda}}$:
\begin{definition}
    A state $\hat{\rho}$ is \textit{Gaussian} if and only if it can be represented as 
    \begin{align}
        \hat{\rho}  = \hat{U}\hat{\rho}_{\vec{\lambda}}\hat{U}^\dagger,\quad   \hat{U}  = \exp\left[i\left(\hat{H}_1+\hat{H}_2\right)\right],
    \end{align}
    where $\hat{H}_1$ and $\hat{H}_2$ are Hermitian linear combinations of $\hat{c}_p$ and $\hat{c}_p\hat{c}_q$ respectively.
\end{definition}
One can easily observe $\hat{H}_1 = 0$ for every even state. Hence, we can simplify this definition when our state is even:
\begin{definition}\label{Two mode even}
    A state $\hat{\rho}$ is \textit{even Gaussian} if and only if it can be represented as 
    \begin{align}
        \hat{\rho}  = \hat{U}\hat{\rho}_{\vec{\lambda}}\hat{U}^\dagger,\quad  \hat{U}  = \exp\left[i\hat{H}_2\right],
    \end{align}
    where $\hat{H}_2$ is a Hermitian linear combination of $\hat{c}_p\hat{c}_q$.
\end{definition}
From this definition, we can derive the Gaussianity of two-mode even states.
\begin{theorem}
    Every two-mode even state, $|\psi\rangle = \alpha|00\rangle+\beta|11\rangle$ is Gaussian.
\end{theorem}
\begin{proof}
    Set $\vec{\lambda} = (1,1)$. Then we have
    \begin{align}
            \hat{\rho}_{\vec{\lambda}} & = \frac{1}{4}\left(\hat{I} + \hat{\sigma}_1^z\otimes \hat{I}_2\right)\left(\hat{I}+\hat{I}_1\otimes\hat{\sigma}_2^z\right) 
             = |00\rangle\langle 00|.
    \end{align}
    Define Hermitian linear combinations of $\hat{c}_p\hat{c}_q$ as
    \begin{align}
    \hat{R}_{\phi} & = \exp\left[-\frac{\phi}{2}\hat{c}_1\hat{c}_2\right] = \exp\left[i\phi\hat{a}_1^\dagger\hat{a}_1\right],\\
            \hat{U}_{\theta} & = \exp\left[-\frac{\theta}{2}\left(\hat{c}_1\hat{c}_4+\hat{c}_2\hat{c}_3\right)\right]  = \exp\left[i\theta\left(\hat{a}_1\hat{a}_2+\hat{a}_2^\dagger\hat{a}_1^\dagger\right)\right],
    \end{align}
where $\theta$ and $\phi$ are real. These operators satisfy the Hermitian condition in definition~\ref{Two mode even},
\begin{align}
    \left(i\frac{\phi}{2}\hat{c}_1\hat{c}_2\right)^\dagger & = \left(\phi\hat{a}_1^\dagger\hat{a}_1\right)^\dagger 
     = \phi\hat{a}_1^\dagger\hat{a}_1, \\
    \left[i\frac{\theta}{2}\left(\hat{c}_1\hat{c}_4+\hat{c}_2\hat{c}_3\right)\right]^\dagger & = \left[\theta\left(\hat{a}_1\hat{a}_2+\hat{a}_2^\dagger\hat{a}_1^\dagger\right)\right]^\dagger \\
    & = \theta\left(\hat{a}_1\hat{a}_2+\hat{a}_2^\dagger\hat{a}_1^\dagger\right).
\end{align}
With simple algebra, we obtain
\begin{align}
        \hat{R}_{\phi}\hat{U}_\theta|00\rangle & = \hat{R}_{\phi}\left\{\cos(\theta)|00\rangle - i\sin(\theta)|11\rangle\right\} \\
        & = \cos(\theta)|00\rangle - i\exp\left(i\phi\right)\sin(\theta)|11\rangle.
\end{align}
It follows that every two-mode even state, $\alpha|00\rangle+\beta|11\rangle$, can be represented as $\hat{R}_{\phi}\hat{U}_\theta|00\rangle$ up to a global phase.
\end{proof}

One can also adapt another method to check the Gaussianity of states, such as Ref.~\cite{Quadrature_Majorana}.

\section{Analysis of time complexity of fermion cases}\label{App:Poly}
In Sec.~\ref{Sec:Fermion}, we depict our focused quantum circuit that allows universal quantum computation in the fermionic system, and investigate how several noises, such as particle loss and dephasing noise, destroy its computational power. Since our algorithm for simulating these noisy cases induces detailed time complexity, we can analytically clarify boundaries that these noisy schemes collapse into a classical simulable regime. In this section, we will discuss the derivation of our proposed boundaries in detail. To achieve this, we start with the following lemma.
\begin{lemma}\label{Poly lemma}
    $x^t = O\left(\textup{poly}(n)\right)$ if $x = 1 + O\left(\log n/t\right)$.
\end{lemma}
\begin{proof}
    Assume that $x = 1 + O\left(\log n/t\right)$. Denote $x = 1+\epsilon$. The assumption implies that there is a constant $C$ such that $\epsilon \leq C\log n/t$. With simple algebra, we compute
    \begin{align}
             x^t & = (1+\epsilon)^t   = \exp\left[t\log\left(1 + \epsilon\right)\right]. 
    \end{align}
    $\log(1+x)\leq x$ holds for every $x>0$, we obtain
    \begin{align}
            t\log(1+\epsilon) & \leq t\epsilon  \leq C\log n.
    \end{align}
    It implies that $x^t = O\left(\text{poly}(n)\right)$.  
\end{proof}
Remark that the computational cost of the particle loss case is given by $O(\nu_{\text{loss}}^{-2f_{\text{loss}}(t,p,\delta)}\text{poly}(n))$,
where $\nu_{\text{loss}}$ and $f_{\text{loss}}(t,p,\delta)$ are defined in Thm.~\ref{Fermion particle loss} and Eq.~\eqref{f1} respectively. For convenience, denote
\begin{align}
        \nu_{\text{loss}}^{-2f_{\text{loss}}(t,p,\delta)} & \leq \left(\frac{3p^2-2p+1}{2p^2-2p+1}\right)^t  \equiv g_1(p)^t.
\end{align}
To reveal the computational boundary, it is sufficient to investigate when $g_1(p)^t=O\left(\text{poly}(n)\right)$.
\begin{theorem}\label{Thm:Fermion_App1}
    Consider the noisy scheme with particle loss, described in Sec.~\ref{Sec:Fermion}. Assume that the number of injected magic states $t$ is given. For $\delta = \Omega\left(1/\textup{poly}(n)\right)$, if the transmission rate $p$ satisfies
    \begin{align}
        p &= O\left(\sqrt{\frac{\log n}{t}}\right).
    \end{align}
    then the time complexity of our algorithm becomes $O(\textup{poly}(n))$.
\end{theorem}
\begin{proof}
    By Lemma~\ref{Poly lemma}, it is sufficient to show $g_1(p) = 1 + O(\log n/t)$. With simple algebra, we can get
    \begin{align}
            g_1(p) - 1 & = \frac{p^2}{2p^2-2p+1}.
    \end{align}
     Since $2p^2-2p+1 \geq\frac{1}{2}$ for all $p\in[0,1]$, the given condition implies 
     \begin{align}\label{Eq:Fermion_App_inequality1}
             \frac{p^2}{2p^2-2p+1} & \leq 2p^2 = O\left(\frac{\log n}{t}\right).
     \end{align}
\end{proof}
\begin{corollary}
    Consider the noisy scheme with particle loss, described in Sec.~\ref{Sec:Fermion}. Assume that the transmission rate $p$ is given. For $\delta = \Omega\left(1/\textup{poly}(n)\right)$, if the number of injected magic states $t$ satisfies
    \begin{align}
        t & = O\left(\frac{\log n}{p^2}\right),
    \end{align}
    then the time complexity becomes $O(\textup{poly}(n))$.
\end{corollary}
\begin{proof}
    The given condition implies that there is a constant $C$ such that
    \begin{align}
        t & \leq \frac{C}{2}\frac{\log n}{p^2},
    \end{align}
    so that 
    \begin{align}
        2p^2 & \leq C \frac{\log n}{t}.
    \end{align}
    Following the proof of Thm.~\ref{Thm:Fermion_App1}, it follows our desired result.
\end{proof}
In a similar way, we can also discuss the dephasing noise case. The computational cost in this case is related to $\nu_{\text{dep}}$ defined in Thm~\ref{Fermion Dephasing}. Remark that the computational cost of the dephasing noise case is given by $ O(\nu_{\text{dep}}^{-2t}\text{poly}(n))$, where $\nu_{\text{dep}}$ is defined in Eq.~\eqref{Eq:nu2}. For convenience, denote
\begin{align}
        \nu_{\text{dep}}^{-2} & = 1+(1-2p)^4   \equiv g_2(p).
\end{align}
Similar to the prior discussion, it is sufficient to investigate when $g_2(p)^t = O\left(\text{poly}(n)\right)$ to reveal the computational boundary. As we mentioned earlier, we restrict the range of $p$ on $[0,1/2]$, because the dephasing noise has a symmetry with the axis $p=1/2$.
\begin{theorem}\label{Thm:Fermion_App2}
    Consider the noisy scheme with dephasing noise, described in Sec.~\ref{Sec:Fermion}. Assume that the number of injected magic states $t$ is given. For $\delta = \Omega\left(1/\textup{poly}(n)\right)$, if the noise rate $p$ satisfies
    \begin{align}
        \frac{1}{2} - p & = O\left(\left(\frac{\log n}{t}\right)^{\frac{1}{4}}\right)
    \end{align} then the time complexity of our algorithm becomes $O(\textup{poly}(n))$.
\end{theorem}
\begin{proof}
    By Lemma~\ref{Poly lemma}, it is sufficient to show $g_2(p) = 1+O(\log n/t)$. Evidently, the condition implies that $g_2(p)= 1+O\left(\log n/t\right)$.
\end{proof}
\begin{corollary}
    Consider the noisy scheme with dephasing noise, described in Sec.~\ref{Sec:Fermion}. Assume that the noise rate $p$ is given. For $\delta = \Omega\left(1/\textup{poly}(n)\right)$, if the number of injected magic states $t$ satisfies
    \begin{align}
        t & = O\left(\frac{\log n}{\left(\frac{1}{2}- p\right)^4}\right),
    \end{align} 
    then the time complexity of our algorithm becomes $O(\textup{poly}(n))$.
\end{corollary}
\begin{proof}
    The given condition implies that there is a constant $C$ such that
    \begin{align}
        t & \leq C\frac{\log n}{\left(\frac{1}{2}- p\right)^4}, 
    \end{align}
    so that $1/2 - p = O\left(\left(\log n/t\right)^{1/4}\right)$. Following the proof of Thm.~\ref{Thm:Fermion_App2}, it follows our desired result.
\end{proof}

\section{Check of Gaussian states}\label{Check Gaussian}
In this section, we verify that the states described from Eq.~\eqref{Eq:Gaussian in Fermion_dephasing} are Gaussian. It is sufficient that these states can be decomposed as tensor products of two-mode even states,
\begin{align}
        \left(\alpha|00\rangle+\beta|11\rangle\right)\left(\gamma|00\rangle+\kappa|11\rangle\right)
\end{align}
up to a global phase~\cite{Gaussian}. Let $\eta = 1-2p$ and 
\begin{align}
    \xi_{\pm} & = \frac{\sqrt{1+\eta^4}\pm \sqrt{1-\eta^4}}{\sqrt{2}\eta^2}.
\end{align}
Then, we note that, for $j=0,1$,
\begin{align}
    |\omega_{\pm}^j\rangle & = \left(\alpha_{\pm}|00\rangle + \beta_{\pm}^j|11\rangle\right)\left(\gamma_{\pm}^j |00\rangle + \kappa_{\pm}|11\rangle\right),
\end{align}
where
\begin{align}
    \alpha_{\pm} & = \left(1 + \xi_{\mp}^2\right)^{-\frac{1}{2}}, \\
    \beta^j_{\pm} & = (-1)^j \xi_{\mp} \left(1 + \xi_{\mp}^2\right)^{-\frac{1}{2}}, \\
    \gamma_{\pm}^j & = (-1)^j(1+\xi_{\pm}^2)^{-\frac{1}{2}}, \\
    \kappa_{\pm} & = \xi_{\pm}(1+\xi_{\pm}^2)^{-\frac{1}{2}}.
\end{align}

\nocite{*}

\end{document}